\documentclass[acmsmall,nonacm]{acmart}

\setcopyright{CC}

\newcommand{\ourtitle}{Chase Termination Beyond Polynomial Time}

\usepackage{subcaption}
\usepackage[inline]{enumitem}
\usepackage[linesnumbered,lined,ruled,commentsnumbered,vlined]{algorithm2e}
\usepackage{tikz}
\usetikzlibrary{calc}
\hypersetup{breaklinks=true}

\usepackage{xspace}
\usepackage{thmtools, thm-restate}

\declaretheorem[sharenumber=dummytheorem]{lemma}
\declaretheorem[sharenumber=dummytheorem,title=Lemma]{applemma} %

\declaretheorem[sharenumber=dummytheorem]{example}
\declaretheorem[sharenumber=dummytheorem]{definition}

\usepackage{xcolor}
\definecolor{darkgreen}{HTML}{36AB14}

\usepackage{todonotes}

\newcommand{\defeq}{:=}

\newcommand{\opfont}[1]{\text{\sf{#1}}} %
\usepackage{bm}
\renewcommand{\vec}[1]{\bm{#1}}
\newcommand{\ltuple}{\langle}
\newcommand{\rtuple}{\rangle}
\newcommand{\tuple}[1]{\ltuple{#1}\rtuple}

\newcommand{\card}[1]{\vert #1 \vert}

\newcommand{\quantify}{\reflectbox{\texttt{Q}}}

\makeatletter
\providecommand*{\cupdot}{%
	\mathbin{%
		\mathpalette\@cupdot{}%
	}%
}
\newcommand*{\@cupdot}[2]{%
	\ooalign{%
		$\m@th#1\cup$\cr
		\hidewidth$\m@th#1\cdot$\hidewidth
	}%
}
\makeatother

\newcommand{\arity}{\opfont{ar}} %
\newcommand{\Dnter}{\mathcal{D}} %
\newcommand{\Inter}{\mathcal{I}} %
\newcommand{\nulls}[1]{\sigNull(#1)} %
\newcommand{\ppos}[2]{\tuple{#1,#2}} %
\newcommand{\arule}{\rho} %
\newcommand{\rulebody}{\opfont{body}} %
\newcommand{\rulehead}{\opfont{head}} %
\newcommand{\aprogram}{\Sigma} %
\newcommand{\adatalogprogram}{\aprogram_{\opfont{DL}}} %

\newcommand{\ledgraph}{\opfont{LXG}} %

\newcommand{\ledgeto}[1]{\stackrel{#1}{\to}} %
\newcommand{\ledge}[3]{#1\ledgeto{#3}#2} %
\newcommand{\lchaseedgeto}[1]{\stackrel{#1}{\twoheadrightarrow}} %
\newcommand{\lchaseedge}[3]{#1\lchaseedgeto{#3}#2} %
\newcommand{\scc}[1]{\opfont{SCC}(#1)} %
\newcommand{\confluence}[1]{\opfont{conf}(#1)} %
\newcommand{\rank}[1]{\opfont{rank}(#1)} %

\newcommand{\fnChase}[1]{\opfont{chase}(#1)} %
\newcommand{\nullvar}[1]{\opfont{var}(#1)} %
\newcommand{\steprule}[1]{\opfont{tgd}[{#1}]} %
\newcommand{\stepmatch}[1]{\sigma[{#1}]} %
\newcommand{\stepmatchex}[1]{\sigma^+[{#1}]} %

\newcommand{\locdec}{\opfont{L}} %
\newcommand{\nufoedge}{\mathrel{\hat{\twoheadrightarrow}}} %
\newcommand{\fanufoedge}{\mathrel{\tilde{\twoheadrightarrow}}} %
\newcommand{\fanufoedgeStar}{\mathrel{\tilde{\twoheadrightarrow}^\ast}} %

\newcommand{\fLineagePath}{\opfont{path}} %

\newcommand{\complclass}[1]{\text{{\sc #1}}\xspace} %

\newcommand{\PTime}{\complclass{P}}

\newcommand{\PSpace}{\complclass{PSpace}}
\newcommand{\ExpSpace}{\complclass{ExpSpace}}
\newcommand{\NExpSpace}{\complclass{NExpSpace}}
\newcommand{\ExpTime}{\complclass{ExpTime}}
\newcommand{\kExpTime}[1]{\ensuremath{#1\text{-\ExpTime}}}
\newcommand{\kExpSpace}[1]{\ensuremath{#1\text{-\ExpSpace}}}
\newcommand{\kNExpSpace}[1]{\ensuremath{#1\text{-\NExpSpace}}}

\newcommand{\natNum}{\mathbb{N}}

\newcommand{\sigPred}{\mathbf{P}}
\newcommand{\sigNull}{\mathbf{N}}
\newcommand{\sigCons}{\mathbf{C}}
\newcommand{\sigVar}{\mathbf{V}}

\newcommand{\predname}[1]{\mathit{#1}\xspace}

\newcommand{\pnElem}{\predname{elem}}
\newcommand{\pnSU}{\predname{su}}
\newcommand{\pnGetSU}{\predname{getSU}}
\newcommand{\pnEmptySet}{\predname{empty}}
\newcommand{\pnSet}{\predname{set}}

\newcommand{\pnNegated}{\predname{neg}}
\newcommand{\pnPositive}{\predname{pos}}
\newcommand{\pnFirst}{\predname{first}}
\newcommand{\pnNext}{\predname{next}}
\newcommand{\pnNextVar}{\predname{nxt}}
\newcommand{\pnLast}{\predname{last}}
\newcommand{\pnLatest}{\predname{new}}

\newcommand{\pnExVar}{\predname{ex}}
\newcommand{\pnOccurs}{\predname{in}}

\newcommand{\pnTrue}{\predname{sat}}
\newcommand{\pnTruePosSucc}{\predname{sat_+}}
\newcommand{\pnTrueNegSucc}{\predname{sat_-}}
\newcommand{\pnAccumulator}{\predname{sat}}

\newcommand{\pnSymb}{\predname{symbol}} %

\newcommand{\pnStep}{\predname{step}}
\newcommand{\pnNextConf}{\predname{to}}
\newcommand{\pnNextConfReal}{\predname{succ}}
\newcommand{\pnTape}{\predname{tape}}
\newcommand{\pnHPos}{\predname{hpos}}
\newcommand{\pnState}{\predname{q}}
\newcommand{\pnAccept}{\predname{acc}}

\newcommand{\symSetPosConjVar}[2]{\textnormal{Pos}^{#1}_{#2}}

\newcommand{\fnApplyPath}{\textnormal{Path}}

\newcommand{\blank}{\texttt{\textvisiblespace}} %

\begin{document}

\title{\ourtitle}
\thanks{This paper is the technical report for the paper with the same title
appearing in the Proceedings of the 43rd Symposium on Principles of Database
Systems (PODS'24) \cite{HanischK:chaseTerminationBeyondPTime:pods24}.}

\author{Philipp Hanisch}
\email{philipp.hanisch1@tu-dresden.de}
\orcid{0000-0003-3115-7492}
\affiliation{%
  \institution{Knowledge-Based Systems Group, TU Dresden}
  \city{Dresden}
  \country{Germany}
}

\author{Markus Krötzsch}
\email{markus.kroetzsch@tu-dresden.de}
\orcid{0000-0002-9172-2601}
\affiliation{%
  \institution{Knowledge-Based Systems Group, TU Dresden}
  \city{Dresden}
  \country{Germany}
}

\begin{abstract}
The \emph{chase} is a widely implemented approach to reason with tuple-generating dependencies (tgds),
used in data exchange, data integration, and ontology-based query answering.
However, it is merely a semi-decision procedure, which may fail to terminate.
Many decidable conditions have been proposed for tgds to ensure chase termination,
typically by forbidding some kind of ``cycle'' in the chase process.
We propose a new criterion that explicitly allows some such
cycles, and yet ensures termination of the standard chase under reasonable conditions.
This leads to new decidable fragments of tgds that are not only
syntactically more general but also strictly more expressive than the fragments 
defined by prior acyclicity conditions.
Indeed, while known terminating fragments are restricted to \complclass{PTime} data complexity,
our conditions yield decidable languages for any \kExpTime{k}.
We further refine our syntactic conditions to obtain fragments of tgds for which an optimised
chase procedure decides query entailment in \PSpace or \kExpSpace{k}, respectively.
\end{abstract}


\maketitle

\section{Introduction}

The \emph{chase} \cite{Benedikt+17:ChaseBench,Maier1979chase,AhoBU79:Chase} is an essential method for
reasoning with constraints in databases, with application areas including
data exchange \cite{FaginKMP05}, constraint implication \cite{BeeriVardi:Chase84},
data cleansing \cite{GMPS14:Llunatic},
query optimization \cite{AhoSU:chaseDepContainment:tods79,Meier:PegasusBackchase:vldb14,BLT14:PDQqueryAnswering},
query answering under constraints \cite{CaliLR:chaseRepairQa:pods03,N+15:RDFoxToolPaper},
and ontological reasoning \cite{BLMS11:decline,CGP12:stickyAIJ}.
The basis for this wide applicability is the chase's ability to compute a
\emph{universal model} for a set of constraints \cite{DNR08:corechase}, which
is either used directly (e.g., as a repaired database) or indirectly (e.g., for deciding query entailment).

However, universal models can be infinite,
and chase termination is undecidable on a single database \cite{BeeriVardi:Chase84} as well as 
in its stricter \emph{uniform} version over arbitrary databases \cite{GM14:ChaseTermUndec,GO:chase18}.
The root of this problem are a form of constraints known as \emph{tuple-generating dependencies} (tgds) 
-- or \emph{existential rules}: Horn logic rules with existential quantification in conclusions --, since they may require additional
domain elements (represented by \emph{nulls}) to be satisfied.

A large body of research is devoted to finding decidable cases for which termination can be guaranteed,
mainly by analysing the data flow, i.e., the propagation of nulls in the chase \cite{DeutschT03,FaginKMP05,Marnette09:superWA,KR11:jointacyc,CG+13:acyclicity,CDK17:rmfa}.%
\footnote{As a notable exception, control flow is analysed in the \emph{graph of rule dependencies} \cite{BLMS11:decline}.}
Here we can distinguish graph-based abstractions, such as \emph{weak acyclicity} \cite{FaginKMP05},
from materialisation-based approaches, such as MFA \cite{CG+13:acyclicity}.
In general, decidable criteria are sufficient but not necessary for termination, but recent breakthroughs 
established decidability of termination for the linear, guarded, and sticky classes of tgds \cite{LMTU19:linear-termination,CalauttiGP:nonuniTermGuardedChase:pods22,CalauttiMP:linskolemchaseterm:pvldb23,GogaczMP:chaseTermGuard:siamjc23}.

Meanwhile, another productive line of recent research clarified the computational power of
the chase, characterising the query functions that can be
expressed by conjunctive queries under tgd constraints.
This expressive power is bounded by data complexity, but can be lower: famously, Datalog
does not capture all queries in \PTime, not even those \emph{closed under homomorphisms}%
\footnote{Such queries are monotone, i.e., intuitively do not need negation (but cf.\ \cite{RT16:DatalogVariants}).} 
\cite{dk08:datalog-lfp}.
Results are more satisfying for tuple-generating dependencies.
Not only do arbitrary tgds capture the recursively enumerable homomorphism-closed queries \cite{RT15ExRulesChar},
but, surprisingly, the decidable homomorphism-closed queries can be captured by tgds for which the standard (a.k.a.\ \emph{restricted}) chase terminates \cite{Bourgaux+:exruleexp:kr21}.
In other words, the standard chase over tgds without any extensions is a universal computational paradigm for database query answering.
This is highly encouraging since a majority of chase implementations already
support this version of the chase procedure \cite{Meier:PegasusBackchase:vldb14,BLT14:PDQqueryAnswering,GMPS14:Llunatic,UKJDC18:VLogSystemDescription,N+15:RDFoxToolPaper,BLMRS15:Graal},
often with favourable performance \cite{Benedikt+17:ChaseBench}.

Naturally, tractable data complexity is often desirable in database applications, and has therefore
been the focus of many works in the area.
Unfortunately, the potential of chase-based computation for more complex computations has 
meanwhile been neglected.
To our knowledge, the only line of research where the chase was used for harder-than-\PTime computations
relies on a method for modelling finite sets with tgds \cite{Carral+19:ChasingSets}.
Practical feasibility was shown for \ExpTime-complete ontology-based query answering \cite{Carral+19:ChasingSets}, 
set-terms in answer set programming \cite{GHK:SetSimASP:IJCAI2022}, complex values in Datalog \cite{MarxKroetzsch:complexValues:ICDT2022}, and
why-provenance \cite{EKM:exruleprov:rr2022}.
In essence, all of these works are based on a single set of tgds for which uniform termination is shown directly.
Known chase termination criteria fail for this case, since they only recognise queries in \PTime.
In fact, most criteria are even known to describe fragments that do not increase upon the expressive power of
Datalog \cite{ZZY15:finitechase,KMR:ChasePower2019},
This limitation is common to all tgd sets on which the semi-oblivious chase uniformly terminates.
We are aware of only one approach so far that studies the more complicated standard chase at all \cite{CDK17:rmfa}, but which also remains in \PTime.

We tackle this challenge with a new method that extends the graph-based termination criterion
of \emph{joint acyclicity} \cite{KR11:jointacyc} with new decidable conditions that allow for some kinds of cycles.
These conditions are specific to the standard chase, and enforce that tgd applications within a cycle are eventually blocked,
possibly only after (double) exponentially many loops. Our conditions detect the uniform termination of the set-modelling tgds
of Carral et al.\ \cite{Carral+19:ChasingSets}, but significantly extend upon this baseline: even a single strongly connected
component in the data-flow graph can correspond to a $\kExpTime{2}$-complete query (whereas Carral et al.\ deal with exponentially many sets).

Leveraging again the graph-based view, we can further analyse the data flows between cyclic strongly connected components to describe decidable
fragments of arbitrary multi-exponential data complexity. Our resulting decidable fragment of \emph{saturating tgds} therefore has non-elementary complexity.
This is already very general, especially when considering that capturing all decidable homomorphism-closed queries
can in principle only be achieved with tgd fragments that are not even recursively enumerable \cite{Bourgaux+:exruleexp:kr21}.

The structure of the data-flow graph enables us to compute more precise $\kExpTime{\kappa}$ bounds for any saturating tgd set.
Refining this analysis further, we then identify cases where the complexity of query answering drops to $\kExpSpace{(\kappa-1)}$.
In particular, we therefore obtain a decidable fragment of uniformly terminating tgd sets with $\PSpace$-complete
query entailment. To the best of our knowledge, this is the first such fragment.

All of our results establish uniform termination of the standard chase under all chase strategies that
prioritise Datalog rules (Kr\"{o}tzsch et al.\ call this the \emph{Datalog-first} strategy \cite{KMR:ChasePower2019}).
By a recent result, this is a stronger requirement than termination under \emph{some} strategies \cite{CLMT:exrulenorm:kr22}. As of today, however,
all known termination criteria imply Datalog-first termination, and preferring Datalog rules is
also a common heuristic in practice.

In summary, our main contributions are as follows:
\begin{itemize}
\item In Section~\ref{sec_dep}, we refine the \emph{dependency graph} of \citeauthor{KR11:jointacyc} \cite{KR11:jointacyc}
to capture data flow in more detail.
\item In Section~\ref{sec_sat}, we study the propagation of inferences between nulls related to strongly connected components in the extended dependency graph.
We find conditions that suffice to prevent infinite repetitions of inference cycles, define the language of \emph{saturating tgds}, and show uniform
standard chase termination for this fragment.
\item In Section~\ref{sec_rank}, we analyse the exact complexity (and size) of the standard chase over saturating sets by assigning \emph{ranks} to strongly connected components in the extended dependency graph. We differentiate the case of single and double exponential complexity for individual components, and we establish matching lower bounds to show that query entailment is $\kExpTime{\kappa}$-complete for tgd sets of rank $\kappa$.
\item In Section~\ref{sec_optchase}, we describe conditions that reduce query entailment for tgd sets of rank $\kappa$ to $\kExpSpace{(\kappa-1)}$. To this end, we discover a tree-like structure within the chase, and we define a syntactic condition called \emph{path guardedness} that allows us to use an optimised chase procedure. Again, we establish matching lower bounds.
\end{itemize}

Detailed proofs are included in the appendix.

\section{Preliminaries}\label{sec_prelims}

We consider a signature based on mutually disjoint, countably infinite sets of \emph{constants} $\sigCons$,
\emph{variables} $\sigVar$, \emph{predicates} $\sigPred$, and \emph{nulls} $\sigNull$.
Each predicate name $p\in\sigPred$ has an \emph{arity} $\arity(p)\geq 0$.
\emph{Terms} are elements of $\sigVar\cup\sigNull\cup\sigCons$.
We use $\vec{t}$ to denote a list $t_1,\ldots,t_{|\vec{t}|}$ of terms, and similarly for special types of terms.
An \emph{atom} is an expression $p(\vec{t})$ with $p\in\sigPred$, $\vec{t}$ a list of terms, and $\arity(p)=|\vec{t}|$.
An \emph{interpretation} $\Inter$ is a set of atoms without variables.
A \emph{database} $\Dnter$ is a finite interpretation without nulls, i.e., a
finite set of \emph{facts} (variable-free, null-free atoms).
For an interpretation $\Inter$, we use $\nulls{\Inter} = \{ n \in \sigNull \mid p(\vec{t}) \in \Inter, n \in \vec{t} \}$ to denote the nulls used in $\Inter$.

\paragraph*{Rules}
A \emph{tuple-generating dependency} (tgd) $\arule$
is a formula
\begin{align}
\arule = \forall \vec{x}, \vec{y}.\, B[\vec{x}, \vec{y}] \to
\exists \vec{v}.\, H[\vec{y}, \vec{v}],\label{eq_rule}
\end{align}
where $B$ and $H$ are
conjunctions of atoms using only terms from $\sigCons$ or from the
mutually disjoint lists of variables $\vec{x}, \vec{y}, \vec{v}\subseteq\sigVar$. 
We call $B$ the \emph{body} (denoted $\rulebody(\arule)$), $H$ the \emph{head} (denoted $\rulehead(\arule)$),
and $\vec{y}$ the \emph{frontier} of $\arule$. %
We may treat conjunctions of atoms as sets, and we omit universal quantifiers in tgds.
We require that all variables in $\vec{y}$ do really occur in $B$ (\emph{safety}). %
A tgd without existential quantifiers is a \emph{Datalog rule}.

\paragraph*{Renamings and Substitutions}
Without loss of generality, we require that variables in tgd sets $\aprogram$ are \emph{renamed apart},
i.e., each variable $x\in\sigVar$ in $\aprogram$ is bound by a unique quantifier in a unique tgd $\arule_x\in\aprogram$.
A \emph{substitution} is a partial mapping $\sigma:\sigVar\to\sigCons\cup\sigVar\cup\sigNull$.
As usual, $x\sigma=\sigma(x)$ if $\sigma$ is defined on $x$, and $x\sigma=x$ otherwise.
For a logical expression $\alpha[\vec{x}]$, the expression $\alpha\sigma$ is obtained by
simultaneously replacing each variable $x$ by $x\sigma$.
Given a list of terms $\vec{t}=t_1,\ldots,t_{|\vec{x}|}$, we write $\alpha[\vec{x}/\vec{t}]$
for the expression $\alpha\{x_1\mapsto t_1,\ldots,x_{|\vec{x}|}\mapsto t_{|\vec{x}|}\}$.
If $\vec{x}$ is clear from the context and no confusion is likely, we write $\alpha[\vec{t}]$ for $\alpha[\vec{x}/\vec{t}]$.

\paragraph*{Semantics}
We consider a standard first-order semantics.
A \emph{match} of a tgd $\arule$ as in \eqref{eq_rule} in an interpretation $\Inter$ 
is a substitution $\sigma$ that maps $\vec{x}\cup\vec{y}$ to terms in $\Inter$,
such that $B\sigma\subseteq\Inter$.
A match is \emph{satisfied} if it can be extended to a substitution $\sigma'$
over $\vec{x}\cup\vec{y}\cup\vec{v}$ such that $H\sigma'\subseteq\Inter$.
A tgd $\arule$ is \emph{satisfied} by $\Inter$, written $\Inter\models\arule$, if
all matches of $\arule$ on $\Inter$ are satisfied.
A fact $\alpha$ is satisfied in $\Inter$, written $\Inter\models\alpha$, if $\alpha\in\Inter$.
Satisfaction extends to sets of tgds and facts as usual.
A tgd or fact $\alpha$ is \emph{entailed} by tgd set $\aprogram$ and database $\Dnter$
if $\Inter\models\alpha$ for all $\Inter$ with $\Inter\models\aprogram\cup\Dnter$. 

\paragraph*{Reasoning with the Chase}
An important reasoning task for tgds is conjunctive query (CQ) answering,
which can further be reduced to the entailment of Boolean CQs (BCQs), which are formulas
$\exists\vec{z}.Q[\vec{z}]$ with $Q$ a conjunction of null-free atoms.
This task is undecidable in general.
A sound and complete (but not always terminating) class of reasoning procedures
is the \emph{chase}, which exists in many variants.
We are interested in the \emph{standard chase} (a.k.a.\ \emph{restricted chase})
under \emph{Datalog-first} strategies.

\begin{definition}\label{def_chase}
A \emph{(standard) chase sequence} for a database $\Dnter$ and a tgd set $\aprogram$ is
a potentially infinite sequence of interpretations $\Dnter^0,\Dnter^1,\ldots$ such that
\begin{enumerate}%
\item \label{i_chase_init}
$\Dnter^0=\Dnter$;
\item \label{i_chase_app}
for every $\Dnter^{i+1}$ with $i\geq 0$, there is a match $\sigma$ for some tgd
$\arule=B[\vec{x},\vec{y}]\to\exists\vec{v}.H[\vec{y},\vec{v}]\in\aprogram$ in $\Dnter^{i}$ such that both of the following hold true:
\begin{enumerate}
\item $\sigma$ is an unsatisfied match in $\Dnter^i$ (i.e., $\sigma$ cannot be extended to a substitution $\sigma^+$ with $H\sigma^+\subseteq\Dnter^{i}$), %
\item $\Dnter^{i+1} = \Dnter^{i} \cup H[\sigma^+(\vec{y}),\sigma^+(\vec{v})]$, where $\sigma^+$ is such that $\sigma^+(y)=\sigma(y)$ for all $y\in\vec{y}$,
      and for all $v\in\vec{v}$, $\sigma^+(v)\in\sigNull$ is a distinct  null not occurring in $\Dnter^{i}$;%
\end{enumerate}
we then say that $\arule$ \emph{was applied} in step $i$, and we define $\steprule{i}\defeq\arule$, $\stepmatch{i}\defeq\sigma$, and $\stepmatchex{i}\defeq\sigma^+$;
\item \label{i_chase_datalogfirst}
if a tgd with existential variables is applied in step $i$, then $\Dnter^i$ must satisfy all Datalog rules in $\aprogram$;%
\item \label{i_chase_fair}
if $\sigma$ is a match for a tgd $\arule\in\aprogram$ and $\Dnter^i$ ($i\geq 0$), then there is $j>i$ such 
that $\sigma$ is satisfied in $\Dnter^j$.%
\end{enumerate}
Item \eqref{i_chase_datalogfirst} requires rule applications to follow a \emph{Datalog-first} strategy, and item \eqref{i_chase_fair} ensures \emph{fairness}.
The \emph{(standard) chase} for such a chase sequence then is $\fnChase{\aprogram,\Dnter}=\bigcup_{i\geq 0}\Dnter^i$.
\end{definition}

A BCQ $q$ is entailed by $\Sigma$ and $\Dnter$ if and only if $\fnChase{\aprogram,\Dnter}\models q$.
If the chase terminates, this can be determined from the (finite) $\fnChase{\aprogram,\Dnter}$.
Termination may depend on the chosen order of tgd applications. 
The Datalog-first strategy \eqref{i_chase_datalogfirst} is a common heuristic that
tends to improve termination in practice, although one can construct examples
where this is not the case \cite{CLMT:exrulenorm:kr22}.
Since Datalog rules can only be applied finitely many times, Datalog-first
does not impair fairness \eqref{i_chase_fair}.

\section{The Labelled Dependency Graph}\label{sec_dep}

The \emph{existential dependency graph} is used to analyse the data flow between
existential variables in a tgd set \cite{KR11:jointacyc}. In particular, a tgd set is
\emph{jointly acyclic} if this graph does not have cycles.
In this section, we recall and slightly extend this approach to better suit our needs. 

The following definition mostly follows \citeauthor{KR11:jointacyc} \shortcite{KR11:jointacyc},
but adds variables as labels to the edges in the graph.
Recall that we assume tgd sets to be renamed apart, so that variables $x$ can be used to identify tgds $\arule_x$.

\begin{definition}\label{def_gex}
	Let $\aprogram$ be a tgd set. A \emph{predicate position} is a pair $\tuple{p,i}\in\sigPred\times\natNum$ with $1\leq i\leq\arity(p)$.
	For a variable $x$ in $\aprogram$, let $\symSetPosConjVar{B}{x}$ (resp.\ $\symSetPosConjVar{H}{x}$) be the set of all predicate positions where $x$
	occurs in the body (resp.\ head) of its unique tgd $\arule_x\in\aprogram$.
	
	For an existential variable $v$ in $\aprogram$, let $\Omega_v$ be the smallest set of positions such that 
	\begin{enumerate*}[label=(\roman*)]
		\item $\symSetPosConjVar{H}{v} \subseteq \Omega_v$, and
		\item for every universal variable $x$, $\symSetPosConjVar{B}{x}\subseteq\Omega_v$ implies $\symSetPosConjVar{H}{x}\subseteq\Omega_v$.
	\end{enumerate*}
	
	The \emph{labelled existential dependency graph} $\ledgraph(\aprogram)$ of $\aprogram$ is a directed graph with the existentially quantified variables of $\aprogram$ as its vertices and an edge $\ledge{v}{w}{y}$ for every tgd $\arule_w\in\aprogram$ with a frontier variable $y$ such that $\symSetPosConjVar{B}{y} \subseteq \Omega_v$.
\end{definition}

\begin{example}\label{ex_dexp}
	This running example illustrates several notions below, hence is not minimal.
	Consider a database $\Dnter_{\predname{lvl}} = \{\predname{first}(1),\predname{last}(\ell)\}\cup \{\predname{next}(i,i+1) \mid 1\leq i<\ell\}$, which defines a total strict order of ``levels'' $1,\ldots,\ell$.
	$\aprogram$ consists of the following tgds with constants $F$ (false) and $T$ (true):
	\begin{align}
		\predname{first}(z)
			&\to \predname{lvl}(F,z)\wedge\predname{lvl}(T,z) \label{eq_dexp_init}\\
		\predname{lvl}(\bar{x}_1,z)\wedge \predname{lvl}(\bar{x}_2,z)
			&\to \exists v.\predname{cat}(\bar{x}_1,\bar{x}_2,z,v) \label{eq_dexp_cat}\\
		\predname{cat}(\bar{x}_1,\bar{x}_2,z,x)
			&\to \predname{part}(\bar{x}_1,x)\wedge \predname{part}(\bar{x}_2,x) \label{eq_dexp_part}\\
		\predname{cat}(\bar{x}_1,\bar{x}_2,z,x)\wedge \predname{next}(z,z_+)
			&\to \exists \bar{w}.\predname{up}(x,z_+,\bar{w}) \label{eq_dexp_up} \\
		\predname{cat}(\bar{x}_1,\bar{x}_2,z,x)\wedge \predname{next}(z,z_+) \wedge
		\predname{up}(x,z_+,\bar{x})
			&\to \predname{lvl}(\bar{x},z_+)\label{eq_dexp_lvl}%
	\end{align}
	Overlined variables denote sequences of $F$ and $T$, where
	$\predname{lvl}(\bar{x},z)$ says that sequence $\bar{x}$ has length $2^z$.
	Initial sequences have length $1$ \eqref{eq_dexp_init}. Longer sequences
	emerge from con\emph{cat}enations \eqref{eq_dexp_cat} with equal-length parts \eqref{eq_dexp_part}.
	The tgd \eqref{eq_dexp_up} promotes concatenations $x$ to sequences $\bar{w}$ on the next level $z_+$.
	We mark such new sequences separately \eqref{eq_dexp_lvl},
	since this will later be useful.
	The tgds \eqref{eq_dexp_init}--\eqref{eq_dexp_lvl} on $\Dnter_{\predname{lvl}}$
	therefore construct all $F$-$T$-sequences of length $2^\ell$, i.e., $2^{2^\ell}$ distinct nulls.
	
	Then $\Omega_v=\{ \ppos{\predname{cat}}{4},\allowbreak \ppos{\predname{part}}{2},\allowbreak\ppos{\predname{up}}{1}\}$ and
	$\Omega_{\bar{w}}=\{\ppos{\predname{up}}{3},\allowbreak\ppos{\predname{lvl}}{1},\allowbreak\ppos{\predname{cat}}{1},\allowbreak\ppos{\predname{cat}}{2}\}$.
	The three edges of $\ledgraph(\aprogram)$ are $v\ledgeto{x\eqref{eq_dexp_up}}\bar{w}$,
	$\bar{w}\ledgeto{\bar{x}_1\eqref{eq_dexp_cat}} v$, and $\bar{w}\ledgeto{\bar{x}_2\eqref{eq_dexp_cat}} v$, where we disambiguate some variables with the numbers of their tgds.
\end{example}

Definition~\ref{def_gex} slightly sharpens the original definition by introducing edges only based on frontier variables.
Moreover, by adding labels, a single edge in the original graph now corresponds to one or more edges in $\ledgraph(\aprogram)$.
Therefore, if the existential dependency graph contains no cycles (i.e., $\aprogram$ is jointly acyclic), then the labelled existential 
dependency graph does not contain cycles either.

$\ledgraph(\aprogram)$ is useful since it over-estimates the possible data flow 
in the computation of $\fnChase{\aprogram,\Dnter}$.
To make this precise, note that any null $n$ in $\fnChase{\aprogram,\Dnter}$ is introduced
in a chase step $i$ as a fresh value $n=\stepmatchex{i}(v)$ of some
existential variable $v$ in $\steprule{i}$.
We write $\nullvar{n}$ for this $v$, and $\lchaseedge{t}{n}{y}$ to indicate that
$\stepmatch{i}(y)=t$ for some term $t\in\sigCons\cup\sigNull$ and frontier variable $y$ of $\steprule{i}$.

\begin{lemma}\label{lemma_ledgraph}
If $\lchaseedge{n_1}{n_2}{y}$ in $\fnChase{\aprogram,\Dnter}$
for nulls $n_1,n_2\in\sigNull$,
then there is an edge $\ledge{\nullvar{n_1}}{\nullvar{n_2}}{y}$ in $\ledgraph(\aprogram)$.
\end{lemma}

\begin{restatable}{lemma}{lemmaInfinitePath}
\label{lemma_infinite_path}
	If $\fnChase{\aprogram, \Dnter}$ is infinite, then
	$\fnChase{\aprogram, \Dnter}$ contains an infinite chain $n_0\lchaseedgeto{y_1} n_1\lchaseedgeto{y_2} \cdots$.
\end{restatable}

By Lemma~\ref{lemma_ledgraph}, the infinite path of Lemma~\ref{lemma_infinite_path} corresponds to an
infinite path in $\ledgraph(\aprogram)$,
which must therefore, since it is finite, contain a cycle.
Hence, if $\ledgraph(\aprogram)$ is acyclic, $\fnChase{\aprogram, \Dnter}$ is finite.

\begin{example}\label{ex_dexp_nonterm}
For $\aprogram$ as in Example~\ref{ex_dexp}, $\ledgraph(\aprogram)$ is cyclic.
Indeed, there are databases $\Dnter$ for which $\fnChase{\aprogram, \Dnter}$ is infinite, e.g.,
for $\Dnter=\{\predname{first}(1),\predname{last}(1),\predname{next}(1,1)\}$.
\end{example}
\section{Termination with Cycles}\label{sec_sat}

In this section, we establish decidable criteria to show
that the chase on a tgd set is guaranteed to be finite, for all input databases, even though the existential dependency graph has
some cycles.

Whenever a tgd $B[\vec{x},\vec{y}]\to\exists\vec{v}.H[\vec{y},\vec{v}]$ was applied in step $i$
of $\fnChase{\Sigma,\Dnter}$, the conjunctive query 
$(B\wedge H)[\vec{x},\vec{y},\vec{v}]$ has the answer $\stepmatchex{i}$ over $\fnChase{\Sigma,\Dnter}$.
Likewise, a chain of chase steps (or path in $\ledgraph(\aprogram)$) leads to a match for a larger query,
defined next.

\begin{definition}\label{def_pathquery}
Consider a path $\mathfrak{p}=v_0\ledgeto{y_1} v_1\ledgeto{y_2} \cdots \ledgeto{y_k} v_k$ in $\ledgraph(\aprogram)$ for variables $v_0,v_i,y_i\in\sigVar$ ($1\leq i\leq k$).
Let $\arule_i: B_i[\vec{x_i},\vec{y_i}]\to\exists\vec{v_i}.H_i[\vec{y_i},\vec{v_i}]$ be a variant of the tgd $\arule_{v_i}$ of variable $v_i$,
where variables have been bijectively renamed such that tgds for different steps do not share variables.
For $1\leq i\leq k$, let $\tilde{y}_i$ (and $\tilde{v}_i$) denote the renamed version of $y_i$ (and $v_i$), and let $\tilde{y}_{k+1}$ denote a fresh variable.

We define $\fnApplyPath(\mathfrak{p}) \defeq \bigwedge_{i=1}^k (B_i\wedge H_i[\tilde{v}_i/\tilde{y}_{i+1}])$.
Given $1\leq i\leq k$, we write
$\fnApplyPath(\mathfrak{p})_{B,i}\defeq B_i$ and $\fnApplyPath(\mathfrak{p})_{H,i}\defeq H_i[\tilde{v}_i/\tilde{y}_{i+1}]$ for the $i$-th body and head conjunction, respectively, with the renamed variables.
\end{definition}

\begin{example}\label{ex_dexp_path}
$\ledgraph(\aprogram)$ of Example~\ref{ex_dexp} has a path
$\mathfrak{p}=v\ledgeto{x} \bar{w}\ledgeto{\bar{x}_1} v$, where we omit the numbers of the tgds from variables for simplicity.
Then $\fnApplyPath(\mathfrak{p}) = \predname{cat}(\bar{x}'_1,\bar{x}'_2,z',x)\wedge \predname{next}(z',z'_+)\wedge \predname{up}(x,z'_+,\bar{x}_1) 
\wedge\predname{lvl}(\bar{x}_1,z'')\wedge \predname{lvl}(\bar{x}''_2,z'')\wedge\predname{cat}(\bar{x}_1,\bar{x}''_2,z'',y_3)$.
Variables marked with ${}'$and ${}''$ stem from renamed variants of tgds \eqref{eq_dexp_cat} and \eqref{eq_dexp_up}, respectively.
Similarly, $\fnApplyPath(\mathfrak{p})_{H,1}=\predname{up}(x,z'_+,\bar{x}_1)$.
\end{example}

\begin{restatable}{lemma}{lemmaPathQuery}
\label{lemma_pathquery}
Let $\mathfrak{c}=n_0\lchaseedgeto{y_1} n_1\lchaseedgeto{y_2} \cdots \lchaseedgeto{y_k} n_k$ be a chain in $\fnChase{\aprogram,\Dnter}$,
with $n_k$ derived in the $m$-th chase step $\Dnter^m$.
There is a path $\mathfrak{p}=\nullvar{v_0}\ledgeto{y_1} \nullvar{n_1}\ledgeto{y_2} \cdots \ledgeto{y_k} \nullvar{n_k}$ in $\ledgraph(\aprogram)$,
and $\Dnter^m\models\fnApplyPath(\mathfrak{p})\theta$ holds for substitution $\theta$
defined as follows:
for each $n_i$ derived in chase step $j$, and each variable $x$ in $\steprule{j}$ that was renamed to $\tilde{x}$ in $\fnApplyPath(\mathfrak{c})$,
$\theta(\tilde{x}) = \stepmatchex{j}(x)$.
\end{restatable}

\begin{figure}[t]
	\begin{tikzpicture}
		\node (v) at (0,0) {$v_0$};
		\node (w1) at (0,-1) {$v_1$};
		\node (wdots) at (0,-2) {\raisebox{0mm}[3.5mm]{$\vdots$}};
		\node (wk) at (0,-3) {$v_\ell$};
		
		\path[draw,->] (v) edge node[right] {$y_1$} (w1);
		\path[draw,->] (w1) edge node[right] {$y_2$} (wdots);
		\path[draw,->] (wdots) edge node[right] {$y_\ell$} (wk);
		\path[draw,->] (wk) edge[out=180, in=180] node[right] {$y_0$} (v);
		
		\node (nk0) at (2,-3) {$n_\ell^0$};
		
		\node (n01) at (4,0) {$n_0^1$};
		\node (n11) at (4,-1) {$n_1^1$};
		\node (n1dots) at (4,-2) {\raisebox{0mm}[3.5mm]{$\vdots$}};
		\node (nk1) at (4,-3) {$n_\ell^1$};

		\path[draw,->>] (nk0) edge[out=0, in=225] (n01);
		\path[draw,->>] (n01) edge (n11);
		\path[draw,->>] (n11) edge (n1dots);
		\path[draw,->>] (n1dots) edge (nk1);
		
		\node (n02) at (6,0) {$n_0^2$};
		\node (n12) at (6,-1) {$n_1^2$};
		\node (n2dots) at (6,-2) {\raisebox{0mm}[3.5mm]{$\vdots$}};
		\node (nk2) at (6,-3) {$n_\ell^2$};

		\path[draw,->>] (nk1) edge[out=0, in=225] (n02);
		\path[draw,->>] (n02) edge (n12);
		\path[draw,->>] (n12) edge (n2dots);
		\path[draw,->>] (n2dots) edge (nk2);
		
		\path[draw,->, dashed,darkgreen,thick] (nk0) edge[out=90, in=180] node[left] {$H_0$} (n01);
		\path[draw,->, dashed,darkgreen,thick] (nk1) edge[out=65, in=-65] node[right] {$H_1$} (n01);
		\path[draw,->, dashed,darkgreen,thick] (nk2) edge[out=65, in=-65] node[right] {$H_2$} (n02);

	\end{tikzpicture}
	\caption{Dependency cycle in $\ledgraph(\aprogram)$ (left) and corresponding chain of derived nulls in $\fnChase{\aprogram,\Dnter}$ (right);
	$H_0$, $H_1$, and $H_2$ denote variants of a tgd head to illustrate propagation}\label{fig_chain_intuition}
\end{figure}
Lemma~\ref{lemma_pathquery} ensures that every chain of nulls in the chase is accompanied by
facts of the form $\fnApplyPath(\mathfrak{p})$, connecting all nulls of the chain.
Figure~\ref{fig_chain_intuition} sketches this situation for a cyclic path (left).
A corresponding chain of nulls $n^0_\ell\lchaseedgeto{y_0}n^1_0\lchaseedgeto{y_1}\ldots{}$ is sketched
on the right, where vertical positions indicate existential variables, i.e., $\nullvar{n^j_i}=v_i$.
Moreover, a solid arrow like $n^0_\ell\twoheadrightarrow n^1_0$ also corresponds to facts of the form
$\fnApplyPath(n^0_\ell\lchaseedgeto{y_0}n^1_0)$ that connect the respective nulls.
The dashed arrows are explained below.

We can use the facts of $\fnApplyPath(\mathfrak{p})$ to infer additional information that may help with
chase termination for cyclic paths.
Indeed, additional information may prevent tgd applications if the head of a tgd is
already entailed (Definition~\ref{def_chase} (2.a)).

To understand this better, let's denote the tgd for %
edge $v_\ell\ledgeto{y_0}v_0$
in Figure~\ref{fig_chain_intuition} as
$\arule=B[\vec{x},\vec{y}]\to\exists\vec{v}.H[\vec{y},\vec{v}]$
with $y_0\in\vec{y}$ and $v_0\in\vec{v}$.
If $i$ is the chase step that introduced $n^1_0$, then $H\stepmatchex{i}\subseteq\Dnter^{i+1}$.
In Figure~\ref{fig_chain_intuition}, $H_0$ represents the facts $H\stepmatchex{i}$,
which involve (among others) the terms $n^0_\ell$ and $n^1_0$, as suggested by the dashed arrow.

Clearly, we cannot apply $\arule$ twice for the same substitution of its frontier:
for all $i'>i$ and substitutions $\sigma$ with $\vec{y}\sigma=\vec{y}\stepmatch{i}$, we have $\Dnter^{i'}\models\exists\vec{v}.H\sigma$
due to the presence of $H_0$.
However, a cycle in $\ledgraph(\aprogram)$ can lead to a chain of tgd applications
as in Figure~\ref{fig_chain_intuition}, where 
$\arule$ is applied to different frontiers in several steps.
Indeed, if $j$ is the chase step that introduced $n^2_0$,
then $y_0\stepmatch{i} \neq y_0\stepmatch{j}$, so the matches differ on at least this frontier variable.
If we write $\vec{y^-}\defeq\vec{y}\setminus\{y_0\}$ for the frontier of $\arule$ without $y_0$,
we can think of $\vec{y^-}\stepmatch{i}$ as the ``context'' for which $\arule$ was applied to $n_\ell^0$ in step $i$.

Our goal is that $\arule$ can never be applied twice to the same context within a single chain.
To ensure this, we would like the chase to derive additional facts
$H_1$ and $H_2$ as indicated in Figure~\ref{fig_chain_intuition}.
Fixing a variable order $H[y_0,\vec{y^-},\vec{v}]$,
these facts are
$H_1\defeq H[n_\ell^1,\vec{y^-}\stepmatch{i},\vec{v}\stepmatchex{i}]$
and 
$H_2\defeq H[n_\ell^2,\vec{y^-}\stepmatch{i},\vec{v}\stepmatchex{j}]$.
Chains like those in Figure~\ref{fig_chain_intuition}, even if finite, can become rather 
long, and we need facts $H_p$ to be propagated to all nulls $n_\ell^p$.
We therefore require two kinds of conditions: (i) a base case that turns recently derived (forward) $H_0$
into a (backwards) $H_1$, and (ii) an inductive step that propagates a (backwards) $H_p$ to
another (backwards) $H_{p+1}$.
The next definition spells out the two conditions.
We generalise slightly by allowing that, instead of applying a single tgd $\arule$
several times, the propagation can occur between different tgds along a path.

\begin{definition}\label{def_propagating}
For $p\in\{\ast,a,b\}$, let $u^p\ledgeto{y^p}v^p\in\ledgraph(\aprogram)$ be
such that the corresponding tgd $\arule^p$ has a head $\exists v^p,\vec{w^p}.H^p[y^p,\vec{z^p},v^p,\vec{w^p}]$.
Moreover, let $\adatalogprogram$ denote the set of Datalog rules in $\Sigma$.

A path $\mathfrak{p}=u^\ast\ledgeto{y^\ast}v^\ast\ledgeto{z_2}\ldots\ledgeto{z_\ell} w$
is \emph{base-propagating} if
\begin{align}
	\adatalogprogram \models \fnApplyPath(\mathfrak{p}) \to H^\ast[\tilde{y}_{\ell+1},\vec{\tilde{z}_1},\tilde{y}_2,\vec{\tilde{w}_1}]\label{eq_propbase}
\end{align}
where the conclusion is the path's first head conjunction
$\fnApplyPath(\mathfrak{p})_{H,1} = H^\ast[\tilde{y}_1,\vec{\tilde{z}_1},\tilde{y}_2,\vec{\tilde{w}_1}]$ with $\tilde{y}_1$ replaced by the variable $\tilde{y}_{\ell+1}$ 
that occurs in $\fnApplyPath(\mathfrak{p})_{H,\ell} = H_\ell[\tilde{y}_\ell,\vec{\tilde{z}_\ell},\tilde{y}_{\ell+1},\vec{\tilde{w}_\ell}]$.

A composite path $\mathfrak{p}$ of the form $\mathfrak{p}^a \mathfrak{p}^b$ with
\[\mathfrak{p}^a=u^a\ledgeto{y^a}v^a\ledgeto{z_2}\ldots\ledgeto{z_\ell} u^b
\text{ and }
\mathfrak{p}^b=u^b\ledgeto{y^b}v^b\ledgeto{z'_2}\ldots\ledgeto{z'_k} w\]
is \emph{step-propagating} for $\arule^\ast$ (or, equivalently, for $u^\ast\ledgeto{y^\ast}v^\ast$) if
\begin{align}
	\adatalogprogram \models \fnApplyPath(\mathfrak{p}) \wedge{} H^\ast[\tilde{y}_{\ell+1},\vec{x_z},\tilde{y}_2,\vec{x_w}] \to H^\ast[\tilde{y}_{\ell+k+1},\vec{x_z},\tilde{y}_{\ell+2},\vec{x_w}]
\label{eq_propstep}
\end{align}
where $\vec{x_z}$ and $\vec{x_w}$ are lists of fresh variables of length $|\vec{x_z}|=|\vec{z^\ast}|$
and $|\vec{x_w}|=|\vec{w^\ast}|$, respectively;
$\tilde{y}_{\ell+k+1}$ is the renamed version of the final variable $w$ in $\fnApplyPath(\mathfrak{p})$;
and the other mentioned variables stem from the atom sets
$\fnApplyPath(\mathfrak{p})_{H,1} = H^a[\tilde{y}_1,\vec{\tilde{z}_1},\tilde{y}_2,\vec{\tilde{w}_1}]$ and
$\fnApplyPath(\mathfrak{p})_{H,\ell+1} = H^b[\tilde{y}_{\ell+1},\vec{\tilde{z}_{\ell+1}},\tilde{y}_{\ell+2},\vec{\tilde{w}_{\ell+1}}]$.
\end{definition}

Intuitively speaking, considering Figure~\ref{fig_chain_intuition},
base propagation allows us to infer $H_1$ from $H_0$, and step propagation allows us to infer $H_2$ from $H_1$.
In contrast to the simplified situation in the figure, the definition clarifies that the propagated head $H^\ast$
is not necessarily the head of the tgds $\arule^a$ and $\arule^b$ that occur in the current piece of chain we
consider.

\begin{example}\label{ex_dexp_propagation}
	We extend $\aprogram$ from Example~\ref{ex_dexp} by two additional tgds:
	\begin{align}
		\predname{up}(y_1,z,\bar{y}_2)\wedge\predname{part}(\bar{y}_2,y_3)
			&\to \predname{up}(y_3,z,\bar{y}_2) \label{eq_dexp_propbase}\\
		\predname{up}(y_3,z,\bar{y}_2)\wedge\predname{part}(\bar{y}_2,y_3) \wedge \predname{up}(y_3,z',\bar{y}_4)\wedge\predname{part}(\bar{y}_4,y_5)
			&\to \predname{up}(y_5,z,\bar{y}_4)\label{eq_dexp_propstep}
	\end{align}
	$\ledgraph(\aprogram)$ remains as in Example~\ref{ex_dexp}.
	The path $\mathfrak{p}=v\ledgeto{x} \bar{w}\ledgeto{\bar{x}_1} v$ from Example~\ref{ex_dexp_path} is base-propagating.
	Indeed, tgds \eqref{eq_dexp_part} and \eqref{eq_dexp_propbase} entail the required tgd
	$\fnApplyPath(\mathfrak{p}) \to \predname{up}(y_3,z'_+,\bar{x}_1)$, in fact even the stronger tgd
	$\predname{up}(x,z'_+,\bar{x}_1)\wedge\predname{cat}(\bar{x}_1,\bar{x}''_2,z'',y_3)\to \predname{up}(y_3,z'_+,\bar{x}_1)$.
	Similarly, for $\mathfrak{p}^a=v\ledgeto{x} \bar{w}\ledgeto{\bar{x}_1} v$ and $\mathfrak{p}^b=v\ledgeto{x} \bar{w}\ledgeto{\bar{x}_2} v$,
	the path $\mathfrak{p}^a\mathfrak{p}^b$ is step-propagating for \eqref{eq_dexp_up} (i.e., for $v\ledgeto{x} \bar{w}$) using
	tgds \eqref{eq_dexp_part} and \eqref{eq_dexp_propstep}.
	Note that variables $y_i$ in \eqref{eq_dexp_propstep} match the path labels in Definition~\ref{def_propagating}.
	
	As we will see below, the extended example achieves universal termination, and in particular also terminates for
	the database of Example~\ref{ex_dexp_nonterm}.
	For the case that $\predname{next}$ is a strict order, tgd \eqref{eq_dexp_lvl} ensures that each sequence is assigned a unique level,
	even with the additional $\predname{up}$-facts from tgds \eqref{eq_dexp_propbase} and \eqref{eq_dexp_propstep}.
\end{example}

The complexity of checking Definition~\ref{def_propagating} is dominated by the \ExpTime complexity of
Datalog.

\begin{restatable}{lemma}{lemmaPropCompl}
\label{lemma_prop_compl}
Checking whether a path $\mathfrak{p}$ is base-propagating (or step-propagating)
is \ExpTime-complete, and \PTime-complete with respect to the length of $\mathfrak{p}$.
\end{restatable}

The conditions of Definition~\ref{def_propagating} have the desired impact on the chase:
if we find a sequence of nulls that corresponds (by Lemma~\ref{lemma_ledgraph}) to a path
that is propagating, then satisfied head atoms in the chase are propagated accordingly.
Moreover, since propagation is defined based on $\adatalogprogram$, the Datalog-first chase
ensures that the propagation happens before further nulls are introduced.
To ensure termination, we require that the cycles in
each strongly connected component in $\ledgraph(\aprogram)$
can be broken by removing a set of edges $E$ that are connected by propagating paths:

\begin{definition}\label{def_esaturating}
Let $C$ be a strongly connected component in $\ledgraph(\aprogram)$, and let $E$
be a set of edges in $C$.
An \emph{$\bar{E}$-path} is a path in $C$ that (i) contains no edges from $E$,
(ii) starts in some $s\in\{v\mid u\ledgeto{y}v\in E\}$, and (iii)
ends in some $t\in\{u\mid u\ledgeto{y}v\in E\}$.
Moreover,
$C$ is \emph{$E$-saturating} if
\begin{enumerate}
\item $C$ without the edges of $E$ is acyclic;\label{item_saturating_acyc}
\item for all $u\ledgeto{x}v, u'\ledgeto{x'}v \in E$, we have $x=x'$; \label{item_saturating_confluent}
\item all paths $e\mathfrak{p}$ in $C$, such that $e\in E$ and $\mathfrak{p}$ is an $\bar{E}$-path,
are base-propagating; and\label{item_saturating_base}
\item all paths $e^a\mathfrak{p}_1 e^b\mathfrak{p}_2$ in $C$, such that $e^a,e^b\in E$ and $\mathfrak{p}_1,\mathfrak{p}_2$ are $\bar{E}$-paths, are step-propagating for every $e\in E$.\label{item_saturating_step}
\end{enumerate}
$\aprogram$ is \emph{saturating} if
all strongly connected components in $\ledgraph(\aprogram)$ are $E$-saturating for some $E$.
\end{definition}

\begin{example}\label{ex_dexp_saturating}
The extended set of tgds from Example~\ref{ex_dexp_propagation} is saturating.
$\ledgraph(\aprogram)$ has a single strongly connected component (cf.\ Example~\ref{ex_dexp}) where
Definition~\ref{def_esaturating} is satisfied for $E=\{v\ledgeto{x}\bar{w}\}$.
For $e = v\ledgeto{x} \bar{w}$, \eqref{item_saturating_base} applies to two paths $e\mathfrak{p}$
and \eqref{item_saturating_step} to four paths $e\mathfrak{p} e\mathfrak{p}'$,
where $\mathfrak{p}, \mathfrak{p}'\in\{\bar{w}\ledgeto{\bar{x}_{i}} v\mid i\in\{1,2\}\}$.
Propagation follows as in Example~\ref{ex_dexp_propagation}.
\end{example}

\begin{restatable}{theorem}{thmSatCompl}
\label{thm_sat_compl}
Deciding if $C$ is $E$-saturating is \ExpTime-complete in the size of $C$.
The same complexity applies to deciding if a tgd set $\aprogram$ is saturating.
\end{restatable}

In practice, the exponential factors in Theorem~\ref{thm_sat_compl} might be small, already because 
strongly connected components are often small.
The next result states that tgds in the edge sets $E$ can only be applied at most once for each substitution of their
``context'' variables -- an essential ingredient for termination.

\begin{restatable}{lemma}{lemmaContextsat}
\label{lemma_contextsat}
Let $C$ be an $E$-saturating strongly connected component in $\ledgraph(\aprogram)$,
and let $u\ledgeto{y}v\in E$ have the corresponding tgd $\arule_v$ with head $\exists v,\vec{w}. H[y,\vec{z},v,\vec{w}]$.
For every database $\Dnter$ and 
chain of nulls $n_0\lchaseedgeto{y}n_1\lchaseedgeto{x_2}\ldots\lchaseedgeto{x_{\ell-1}}n_{\ell-1}\lchaseedgeto{y}n_\ell$ in $\fnChase{\aprogram,\Dnter}$ such that $n_0=y\stepmatch{a}$, $n_1=v\stepmatchex{a}$, $n_{\ell-1}=y\stepmatch{b}$, and $n_\ell=v\stepmatchex{b}$ for chase steps $a<b$, we have $\vec{z}\stepmatch{a}\neq\vec{z}\stepmatch{b}$.
\end{restatable}

The main result of this section is as follows. 
A more detailed analysis follows in the next section.

\begin{restatable}{theorem}{thmSaturTerm}
\label{thm_satur_term}
If $\aprogram$ is saturating,
then $\fnChase{\aprogram,\Dnter}$ is finite for all databases $\Dnter$. 
\end{restatable}

\begin{example}\label{ex_sets_saturating}
	Theorem~\ref{thm_satur_term} also subsumes previous termination results by \citeauthor{Carral+19:ChasingSets} \cite{Carral+19:ChasingSets}.
	The following tgds captures the essence of their approach of simulating finite sets in tgds:
	\newcommand{\shortwedge}{\,{\wedge}\,}
	\begin{align} 
		\pnElem(x) \shortwedge \pnSet(S) & \to \exists v.\, \pnSet(S)\shortwedge \pnSU(x, S, v) \shortwedge \pnSU(x, v, v) \label{eq_set_construct}\\
		\!\pnSU(x, S, T) \shortwedge \pnSU(y, S, S) & \to \pnSU(y, T, T) \label{eq_set_propagate}%
	\end{align}
	The database can provide $\pnElem$-facts that define a domain of elements, and a fact $\pnSet(\emptyset)$.
	The tgd~\eqref{eq_set_construct} constructs new ``sets'' by creating facts $\pnSU(x,S,T)$, which can be read as $\{x\}\cup S=T$.
	In particular, $\pnSU(x,S,S)$ means $x \in S$.
	The tgd~\eqref{eq_set_propagate} propagates memberships $x \in S$ from $S$ to a direct superset $T$, 
	which we recognise as a special case of step propagation.
	Indeed, the only dependency here is $v\ledgeto{S}v$, and all paths $\mathfrak{p}_{(1/2)}$ in Definition~\ref{def_esaturating}
	are empty. Base propagation is achieved by the atom $\pnSU(x, v, v)$ in \eqref{eq_set_construct}
	without requiring any Datalog rules.
\end{example}

\begin{example}\label{ex_e_must_be_confluent}
Definition~\ref{def_esaturating} \eqref{item_saturating_confluent} excludes ``confluent''
edges with distinct labels: being based on single frontier variables, our propagation
conditions do not suffice for this case.
For example, consider the tgd set $\aprogram$ that extend the tgd set of Example~\ref{ex_sets_saturating}:
\begin{align}
	\pnSet(S) &\to \pnElem(S) \label{eq_confluent_edges_trigger}\\
	\pnSU(x, S, T) &\to \pnSU(T, S, T) \label{eq_confluent_edges_base_prop}\\
	\pnSU(x, S, T) \land \pnSU(y, x, x) &\to \pnSU(y, T, T) \label{eq_confluent_edges_step_prop_1}\\
	\pnSU(x, S, T) \land \pnSU(S, y, S) &\to \pnSU(T, y, T) \label{eq_confluent_edges_step_prop_2}\\
	\pnSU(x, S, T) \land \pnSU(x, y, x) &\to \pnSU(T, y, T) \label{eq_confluent_edges_step_prop_3}
\end{align}
We allow ``sets'' to be used as ``elements'' \eqref{eq_confluent_edges_trigger}, thereby creating a confluent $E$-edge.
The tgd~\eqref{eq_confluent_edges_base_prop} ensures that the new $E$-edge is base-propagating, and tgds~\eqref{eq_confluent_edges_step_prop_1} – \eqref{eq_confluent_edges_step_prop_3} ensure that the additional pairs of $E$-edges are step-propagating.
For $\Dnter = \{ \pnSet(\emptyset) \}$, $\fnChase{\aprogram, \Dnter}$ is infinite.

However, strongly connected components may have arbitrary confluent edges that are not 
in $E$, as in Example~\ref{ex_dexp_saturating}, and $E$ itself may have confluent edges
that are using the same label.
\end{example}

\section{Complexity of the Saturating Chase}\label{sec_rank}

Next, we refine Theorem~\ref{thm_satur_term} by deriving specific bounds for the size
of the chase over saturating tgd sets $\aprogram$, based on the structure of $\ledgraph(\aprogram)$.
For a vertex $v$ of $\ledgraph(\aprogram)$, we write $\scc{v}$ for the strongly connected component
that contains $v$, and $\scc{\ledgraph(\aprogram)}$ for the set of all strongly connected components.
An edge $u \ledgeto{y} w$ is \emph{incoming} for $\scc{v}$ if $u\notin\scc{v}$ and $w\in\scc{v}$;
in this case we write $\scc{u}\prec\scc{w}$ (where $\scc{w}=\scc{v}$).
The transitive reflexive closure of $\prec$ is the usual induced partial order on $\scc{\ledgraph(\aprogram)}$.

\begin{definition}\label{def_confluence}
For vertex $v$ of $\ledgraph(\aprogram)$, we define
the label set $\lambda(v) = \{ x \mid u \in \scc{v}, u \ledgeto{x} v\text{ in }\ledgraph(\aprogram)\}$,
and \emph{confluence} $\confluence{v} = \card{ \lambda(v) }$.
For $C\in\scc{\ledgraph(\aprogram)}$, we define
$\confluence{C} = \max \{ \confluence{u} \,{\mid}\, u\,{\in}\, C \}$.
$C$ is \emph{homogeneously confluent} if $\confluence{C} \,{=}\, 1$.
\end{definition}

Note that $\confluence{C} = 0$ implies that $C$ is trivial, i.e., a singleton set without any cycle (self loop).

\begin{definition}\label{def_C_input_depth}
Consider a database $\Dnter$ and $C\in\scc{\ledgraph(\aprogram)}$.
A term $t$ in $\fnChase{\aprogram, \Dnter}$ is a \emph{$C$-input}
if (i) $t$ is a constant, or (ii) $t$ is a null and $C$ has an incoming edge $\nullvar{t}\ledgeto{x} w$.
A null $n$ with $\nullvar{n}\in C$ has \emph{$C$-depth} $k$ if there is a maximal chain of nulls
$m_0\lchaseedgeto{y_1}\ldots\lchaseedgeto{y_k}m_k$ in $\fnChase{\aprogram, \Dnter}$ with
$m_k = n$ and $\nullvar{m_i} \in C$ for $0 \leq i \leq k$.
The $C$-depth of $n$ is undefined if the length of such chains is unbounded.
\end{definition}

The next result limits the number of bounded-$C$-depth nulls based on the number of 
$C$-inputs, thereby also clarifying the significance of homogeneous confluence.
Note that this general insight does not restrict to saturating tgd sets.

\begin{restatable}{lemma}{lemmaNullsAtDepth}
\label{lemma_nulls_at_depth}
Consider a database $\Dnter$ and $C\in\scc{\ledgraph(\aprogram)}$.
If $i$ is the number of $C$-inputs in $\fnChase{\aprogram, \Dnter}$,
then, for any $d\geq 0$, the number of nulls at $C$-depth $\leq d$ is at most
doubly exponential in $d$ and polynomial in $i$.
If $C$ is homogeneously confluent, then
this number is at most exponential in $d$.
\end{restatable}

\begin{definition}\label{def_rank_comp}	
Let $\aprogram$ be saturating with $E_C$ denoting the set
$E$ of Definition~\ref{def_esaturating} for $C\in\scc{\ledgraph(\aprogram)}$.
Let $C_1,\ldots,C_k$ be a list of all $C_i\in\scc{\ledgraph(\aprogram)}$ topologically ordered w.r.t.\ the induced order $\prec$.

We define $\rank{C_i}$ 
iteratively for $i=1,\ldots,k$ as follows,
where we assume $\max\{\}=0$.
First, let $r_{\opfont{in}}^i \defeq \max\{\rank{C_j}\mid C_j\prec C_i\}$ and let
$r_{\opfont{cxt}}^i \defeq \max\{\rank{C}\mid C\in \mathcal{C}_\opfont{cxt}^i\}$ where
$\mathcal{C}_\opfont{cxt}^i$ denotes the set
\[
\{C_j\mid v,w\in C_i, u\in C_j, w\ledgeto{x}v \in E_{C_i}, u\ledgeto{y}v\in \ledgraph(\aprogram), x\neq y\}.
\]
Then the rank of $C_i$ is
\begin{align*}
\rank{C_i} \defeq
	\begin{cases} 
		r_{\opfont{in}}^i & \text{if }\confluence{C_i} = 0, \\
		\max\{ r_{\opfont{in}}^i, r_{\opfont{cxt}}^i + 1 \} & \text{if }\confluence{C_i} = 1, \\
		\max\{ r_{\opfont{in}}^i, r_{\opfont{cxt}}^i + 2 \} & \text{if }\confluence{C_i} \geq 2.
	\end{cases}
\end{align*}
The rank of $\aprogram$ is $\rank{\aprogram}=\max\{\rank{C}\mid C\in\scc{\ledgraph(\aprogram)}\}$.
\end{definition}

\begin{restatable}{theorem}{thmBoundedNulls}
\label{thm_bounded_nulls}
Let $\aprogram$ be saturating and $\Dnter$ a database. 
For %
every existential variable $v$ in $\aprogram$, the
number of nulls $n$ with $\nullvar{n} = v$ in $\fnChase{\aprogram, \Dnter}$
is at most $\rank{\scc{v}}$-exponential in the size of $\Dnter$.
\end{restatable}

Example~\ref{ex_dexp_saturating} (and the discussion in Example~\ref{ex_dexp})
shows that the upper bound of Theorem~\ref{thm_bounded_nulls} can be reached. We can further strengthen this into a hardness result:

\begin{restatable}{theorem}{thmRankCompl}
\label{theo_rank_compl}
Let $\aprogram$ be saturating.
For every database $\Dnter$ the size of $\fnChase{\aprogram, \Dnter}$ is at most $\rank{\aprogram}$-exponential in the size of $\Dnter$,
and BCQ entailment is \kExpTime{\rank{\aprogram}}-complete for data complexity.
\end{restatable}

\section{The Chase in the Forest}\label{sec_optchase}

In this section, we refine Theorem~\ref{theo_rank_compl} by identifying cases where BCQ answering over saturating tgd sets $\aprogram$
is not $\kExpTime{\rank{\aprogram}}$-complete but merely $\kExpSpace{(\rank{\aprogram}-1)}$-complete, and we design a chase procedure
that runs within these complexity bounds.
To simplify presentation, we consider tgd sets with a single rank-maximal SCC $\hat{C}$ in $\ledgraph(\aprogram)$
(generalisations are possible; see concluding remarks).
If $\confluence{\hat{C}} = 1$, we can establish a tree-like search space within the chase that follows the edges of $\ledgraph(\aprogram)$.
A new syntactic restriction on tgds, called \emph{path guardedness}, ensures that a chase that follows this
tree-like structure remains complete for conjunctive query answering.

\begin{definition}\label{def_arboreous}
A tgd set $\aprogram$ is \emph{arboreous} if it is saturating, and has a unique $\hat{C}\in\scc{\aprogram}$
with $\rank{\hat{C}}=\rank{\aprogram}$, which satisfies $\confluence{\hat{C}} \leq 1$.
For such $\aprogram$ and some $\fnChase{\aprogram, \Dnter}$, the \emph{null forest} is the directed graph $\tuple{\hat{N},\nufoedge}$ with
$\hat{N}=\{n\in\sigNull(\fnChase{\aprogram, \Dnter})\mid\nullvar{n}\in\hat{C}\}$ the nulls of variables in $\hat{C}$,
and $n\nufoedge m$ if $n\lchaseedgeto{x}m\text{ for some }x$.
\end{definition}

\begin{restatable}{lemma}{lemmaArboreous}
\label{lemma_arboreous}
For every arboreous $\aprogram$ and $\fnChase{\aprogram, \Dnter}$, the null forest is indeed a forest (set of trees).
\end{restatable}

Next, we use the special tgds that correspond to set $E$ of Definition~\ref{def_esaturating}
to partition the null forest into sub-forests. The intuition is that special tgd applications
start a new sub-forest (their fresh nulls being the roots), whereas other tgd applications remain within
their current sub-forest. By placing all remaining terms of the chase in an additional root node, we obtain a
tree structure whose nodes are sets of terms that partition the terms of the chase:

\begin{definition}\label{def_fac_null_forest}
Let $\aprogram$, $\hat{C}$, $\hat{N}$, and $\nufoedge$ be as in Definition~\ref{def_arboreous},
let $\hat{E}$ be the edge set $E$ of Definition~\ref{def_esaturating} for $\hat{C}$, and let 
$T$ be the set of all terms in $\fnChase{\aprogram, \Dnter}$.
An \emph{$\hat{E}$-tgd} is a tgd with an existential variable that is the target of an edge in $\hat{E}$.
The \emph{$\hat{E}$-variables} $V_{\hat{E}}$ are all existential variables in $\hat{C}$ that occur in some
$\hat{E}$-tgd.

For every chase step $i$ where an $\hat{E}$-tgd was applied, let
$R[i]$ be the set of fresh nulls introduced in step $i$. 
Let $\bar{R}\subseteq\hat{N}$ be the set of nulls that are not in any such $R[i]$.
Then, for each $R[i]$, let $F[i]\subseteq\hat{N}$ be the least set that contains $R[i]$ and
all $m\in\bar{R}$ for which there is $n\in F[i]$ with $n\nufoedge m$.
Let $\mathcal{F}$ be the set of all such sets $F[i]$.
For $F_1,F_2\in\mathcal{F}$ with $F_1\neq F_2$, we write $F_1\fanufoedge F_2$ if there are 
$n_i\in F_i$ ($i=1,2$) such that $n_1\nufoedge n_2$.

Let $\locdec_0=T\setminus\bigcup_{F\in\mathcal{F}} F$.
The \emph{term tree of $\fnChase{\aprogram, \Dnter}$} is the graph $\tuple{N_\sim,\fanufoedge}$ with
$N_\sim = \mathcal{F}\cup\{\locdec_0\}$, and $\fanufoedge$ extended to $N_\sim$ by setting
$\locdec_0\fanufoedge F$ for all $F\in\mathcal{F}$ that have no $\fanufoedge$-predecessor in $\mathcal{F}$.
The reflexive transitive closure of $\fanufoedge$ is denoted $\fanufoedgeStar$.
\end{definition}

\begin{restatable}{lemma}{lemmaFacNullForest}
\label{lemma_fac_null_forest}
For every arboreous $\aprogram$ and $\fnChase{\aprogram, \Dnter}$,
$N_\sim$ is a partition of the terms in $\fnChase{\aprogram, \Dnter}$,
and the term tree is indeed a tree.
\end{restatable}

For every term $t$, Lemma~\ref{lemma_fac_null_forest} allows us to use $\locdec(t)$ to denote the
unique set $\locdec\in N_\sim$ with $t\in\locdec$.

\begin{figure}[t]
	\begin{tikzpicture}
		\node (v) at (0,0) {$v_0$};
		\node (w1) at (0,-1) {$v_1$};
		\node (wdots) at (0,-2) {\raisebox{0mm}[3.5mm]{$\vdots$}};
		\node (wk) at (0,-3) {$v_\ell$};
		
		\path[draw,->] (v) edge node[right] {$y_1$} (w1);
		\path[draw,->] (w1) edge node[right] {$y_2$} (wdots);
		\path[draw,->] (wdots) edge node[right] {$y_\ell$} (wk);
		\path[draw,->] (wk) edge[out=180, in=180] node[right] {$y_0$} (v);
		
		\node (nk0) at (2,-3) {$n_\ell^0$};
		
		\draw[dotted,thick,rounded corners] ($(nk0.north west)+(-0.25,0.25)$) rectangle ($(nk0.south east)+(0.25,-0.25)$);
		
		\node (n01) at (4,0) {$n_0^1$};
		\node (n11) at (4,-1) {$n_1^1$};
		\node (n1dots) at (4,-2) {\raisebox{0mm}[3.5mm]{$\vdots$}};
		\node (nk1) at (4,-3) {$n_\ell^1$};

		\path[draw,->>] (nk0) edge[out=0, in=225] (n01);
		\path[draw,->>] (n01) edge (n11);
		\path[draw,->>] (n11) edge (n1dots);
		\path[draw,->>] (n1dots) edge (nk1);
		\draw[dotted,thick,rounded corners] ($(n01.north west)+(-0.25,0.25)$) rectangle ($(nk1.south east)+(0.25,-0.25)$);
		
		\node (n02) at (6,0) {$n_0^2$};
		\node (n12) at (6,-1) {$n_1^2$};
		\node (n2dots) at (6,-2) {\raisebox{0mm}[3.5mm]{$\vdots$}};
		\node (nk2) at (6,-3) {$n_\ell^2$};

		\path[draw,->>] (nk1) edge[out=0, in=225] (n02);
		\path[draw,->>] (n02) edge (n12);
		\path[draw,->>] (n12) edge (n2dots);
		\path[draw,->>] (n2dots) edge (nk2);
		\draw[dotted,thick,rounded corners] ($(n02.north west)+(-0.25,0.25)$) rectangle ($(nk2.south east)+(0.25,-0.25)$);
		
		\path[draw,->, dashed,darkgreen,thick] (nk0) edge[out=90, in=180] node[left] {$H_0$} (n01);
		\path[draw,->, dashed,darkgreen,thick] (nk1) edge[out=65, in=-65] node[right] {$H_1$} (n01);
		\path[draw,->, dashed,darkgreen,thick] (nk2) edge[out=65, in=-65] node[right] {$H_2$} (n02);
	\end{tikzpicture}
	\caption{Example graph $\ledgraph(\aprogram)$ (left) and chain of derived nulls in $\fnChase{\aprogram,\Dnter}$ (right) from Figure~\ref{fig_chain_intuition}, with nulls partitioned according to Definition~\ref{def_fac_null_forest} for the set $\hat{E} = \{ v_\ell\ledgeto{y_0} v_0 \}$
	}\label{fig_chain_partition_intuition}
\end{figure}
\begin{example}
Figure~\ref{fig_chain_partition_intuition} revisits the abstract example from Figure~\ref{fig_chain_intuition}, which we assume to be saturating according to
Definition~\ref{def_esaturating} using $E=\{ v_\ell\ledgeto{y_0} v_0 \}$.
If this is the unique maximal SCC $\hat{C}$, then $\hat{E}=E$, and we obtain three partitions of nulls that are illustrated in  Figure~\ref{fig_chain_partition_intuition}:
$\locdec(n_\ell^0)$, $\locdec(n_0^1)=\locdec(n_1^1)=\dots=\locdec(n_\ell^1)$, and $\locdec(n_0^2)=\locdec(n_1^2)=\dots=\locdec(n_\ell^2)$.
These partitions are part of a path in the term tree, which is overlaid on the original null forest.

The motivation for defining such a coarser tree structure is that we intend to use this tree to guide the computation of facts during the chase.
We will limit its space-complexity by storing, at each particular moment during the chase, only those facts that can be represented using terms on
a single path of this tree. The coarser the tree structure, the more facts can be considered at any moment, the more cases can be handled by this limited form of chase.

Besides this general intuition, the factorisation also plays a crucial role in deriving a syntactic criterion to recognise
cases where such a tree-based chase can safely be applied, which we will consider next.
The difficulty for this endeavour is that any such syntactic condition eventually has to rely on the facts that have induced the
tree-like structure in the first place, such as $H_0$ in Figure~\ref{fig_chain_partition_intuition}. But these very facts also 
occur in backwards direction to ensure saturation, as indicated by $H_1$ and $H_2$ in the figure.
In the coarser tree structure, such ``backward edges'' merely lead to the same node of the tree, rather than to a predecessor node from which
we could enter parallel branches in forward direction.
\end{example}

The tree structure of terms as such does not constrain the structure
of inferred facts in $\fnChase{\aprogram, \Dnter}$, which may relate nulls from
arbitrary positions in the null forest.
We seek syntactic restrictions that ensure that the chase respects the term tree in
the sense that the terms of any fact are on a common path and impose an order on terms that matches their position in the path.
To this end, we derive relationships $\ppos{p}{i}\preceq \ppos{p}{j}$ on predicate positions such that,
for all $p(\vec{t})\in \fnChase{\aprogram, \Dnter}$ with $t_i,t_j\in\hat{N}$,
we have that $\locdec(t_i)$ is an ancestor of (or possibly equal to) $\locdec(t_j)$ 
in the term tree.
The next definition once again uses our assumption that distinct tgds do not share variables.

\begin{definition}\label{def_constraints_new}
Let $\aprogram$ be arboreous with $\hat{C}$, $\hat{E}$, and $V_{\hat{E}}$ as in Definition~\ref{def_fac_null_forest}.
We will define a (not necessarily transitive) binary relation $\preceq$ on predicate positions $\ppos{p}{i}$.
Any such $\preceq$ induces a relation $\trianglelefteq$ on variables of $\aprogram$ as the reflexive, transitive
closure of the set of all $x\trianglelefteq y$ such that $x$ and $y$ occur at positions $\ppos{p}{i}$ and $\ppos{p}{j}$
in a single body atom of some tgd in $\aprogram$, and $\ppos{p}{i}\preceq\ppos{p}{j}$.

Now, concretely, let $\preceq$ denote the \emph{largest} relation on predicate positions such that,
for all head atoms $p(\vec{t})$ in tgds of $\aprogram$ with universally quantified variables $\vec{y}$
and existentially quantified variables $\vec{v}$:
\begin{enumerate}
\item if $t_i\in\vec{v}\cap \hat{C}$ and $t_j\in\vec{v}\setminus\hat{C}$ then $\ppos{p}{i}\not\preceq\ppos{p}{j}$,\label{item_constraints_root}
\item if $t_i\in\vec{v}\cap V_{\hat{E}}$ and $t_j\in\vec{y}$ then $\ppos{p}{i}\not\preceq\ppos{p}{j}$,\label{item_constraints_newbag}
\item if $t_i\in\vec{v}\cap\hat{C}$ and $t_j\in\vec{y}$ with $t_j\notin\lambda(t_i)$ then $\ppos{p}{i}\not\preceq\ppos{p}{j}$,\label{item_constraints_newstart}
\item if $t_i,t_j\in\vec{y}$ and $t_i\ntrianglelefteq t_j$  then $\ppos{p}{i}\not\preceq\ppos{p}{j}$.\label{item_constraints_propagation}
\end{enumerate}
\end{definition}

One can construct $\preceq$ in polynomial time with a simple greatest fixed point computation.
Note that such a construction is anti-monotonic in $\aprogram$: more tgds lead to fewer constraints $\preceq$.

\begin{restatable}{lemma}{lemmaConstraintsNew}
\label{lemma_constraints_new}
	Let $\aprogram$ be arboreous.
	For every $\Dnter$, if $p(t_1,\ldots,t_n)\in\fnChase{\aprogram, \Dnter}$ with $\ppos{p}{i}\preceq\ppos{p}{j}$,
	then $\locdec(t_i)\fanufoedgeStar\locdec(t_j)$.
\end{restatable}

\begin{definition}\label{def_path_guarded}
	Let $\aprogram$ be arboreous with $\hat{C}$ and $\trianglelefteq$ as in Definition~\ref{def_constraints_new}.
	The set of \emph{$\hat{C}$-affected} positions $\hat{\Omega}$ in $\aprogram$ is $\bigcup_{v\in\hat{C}}\Omega_v$, with $\Omega_v$ as in
	Definition~\ref{def_gex}. A body variable $x$ is \emph{$\hat{C}$-affected} if $\symSetPosConjVar{B}{x} \subseteq \hat{\Omega}$.
	
	A tgd $\arule\in\aprogram$ is \emph{path-guarded} if all $\hat{C}$-affected body variables in $\arule$ are 
	mutually comparable with respect to the relation $\trianglelefteq$, i.e., form a chain in $\trianglelefteq$.
	$\aprogram$ is \emph{path-guarded} if all of its tgds are.
\end{definition}

Any node $\locdec$ of the term tree induces a unique upwards path $\fLineagePath(\locdec)$ that
consists of all nodes $\locdec'$ with $\locdec'\fanufoedgeStar\locdec$. For term $t$, we write $\fLineagePath(t)$ for $\fLineagePath(\locdec(t))$.
Inferences of path-guarded tgds are situated on such paths:

\begin{restatable}{lemma}{lemmaLocalRuleApplication}
\label{lemma_local_rule_application}
Let $\aprogram$ be arboreous and path-guarded, and $\Dnter$ some database.
For every step $i$ in $\fnChase{\aprogram, \Dnter}$,
with $\steprule{i}=B[\vec{x},\vec{y}] \to\exists\vec{v}. H[\vec{y},\vec{v}]$,
and $T_B=(\vec{x}\cup\vec{y})\stepmatch{i}$ and $T_H=(\vec{y}\cup\vec{v})\stepmatchex{i}$:
\begin{enumerate}
\item if $T_B\neq\emptyset$ then $T_B\subseteq\bigcup\fLineagePath(t)$ for some $t\in T_B$, and\label{item_loc_rule_body}
\item if $T_H\neq\emptyset$ then $T_H\subseteq\bigcup\fLineagePath(t)$ for some $t\in T_H$.\label{item_loc_rule_head}
\end{enumerate}
\end{restatable}

\newcommand{\lcRef}[1]{\hyperref[#1]{(L\ref*{#1})}} %
\begin{algorithm}[t]
	\SetKwInput{KwIn}{In}
	\SetKwInput{KwOut}{Out}
	\SetKw{KwBreak}{break}
	\SetKw{KwChoose}{choose}
	\SetKw{KwAnd}{$\;$and$\;$}
	\SetKw{KwOr}{$\;$or$\;$}
	
	\KwIn{tgd set $\aprogram$, database $\Dnter$, BCQ $q$, variable set $V_{\hat{E}}$, integer $M$}
	\KwOut{$\aprogram \cup \Dnter \models q$}
	
	\BlankLine
	$\Inter \coloneqq \Dnter$ \label{l_chase_iinit}\\
	$\mathcal{T} \coloneqq \tuple{\sigCons_0}$ with $C_0$ the set of all constants in $\aprogram$ and $\Dnter$ \label{l_chase_tinit}\\
	\For{$i\in\{1,\ldots,\card{q}\}$}{\label{l_chase_qloop}
		\For{$j\in\{1,\ldots,M\}$}{\label{l_chase_iloop}
			\KwChoose $\arule\,{\in}\,\aprogram$ with frontier $\vec{y}$, and match $\sigma$ applicable %
			to $\Inter$ in Datalog-first chase
			\KwOr \KwBreak (go to \lcRef{l_chase_qloop})\\\label{l_chase_choose}
			\While{$\card{\mathcal{T}} > 1$\KwAnd$\vec{y}\sigma \cap \mathcal{T}.\mathtt{last()} = \emptyset$}{\label{l_chase_prunet}
				$\mathcal{T}.\mathtt{pop()}$
			}
			\eIf{$\vec{v} \cap V_{\hat{E}} \neq \emptyset$}{
				$\mathcal{T}.\mathtt{push}(\vec{v}\sigma^+)$ \label{l_chase_add_to_new_bag}
			}{
				$\mathcal{T}.\mathtt{push}(\mathcal{T}.\mathtt{pop()} \cup \vec{v}\sigma^+)$ \label{l_chase_add_to_last_bag}
			}
			$\Inter \coloneqq \{ p(\vec{t}) \in (\Inter \cup \rulehead(\arule)\sigma^+) \mid \vec{t} \subseteq \bigcup_{T \in \mathcal{T}}T \}$\\\label{l_chase_update_i}
		}
		$\mathcal{T} \coloneqq \tuple{\bigcup_{T \in \mathcal{T}}T}$\label{l_chase_troot}
	}
	
	\KwRet{$\Inter \models q$}\label{l_chase_return}
	
	\caption{Non-deterministic chase for BCQ entailment}
	\label{algo_pg_chase}
\end{algorithm}
Algorithm~\ref{algo_pg_chase} specifies a non-deterministic chase procedure to check the entailment of a BCQ.
It is intended for arboreous, path-guarded tgd sets $\aprogram$ with variables $V_{\hat{E}}$ as in Definition~\ref{def_fac_null_forest}.
The input $M$ bounds the length of the search: we will determine a $\rank{\aprogram}$-exponential value for $M$ for which the
algorithm decides query entailment in $\kNExpSpace{(\rank{\aprogram}-1)}$, showing the problem to be in
$\kExpSpace{(\rank{\aprogram}-1)}$ by Savitch's Theorem.

Algorithm~\ref{algo_pg_chase} maintains a set of inferences $\Inter$ over terms in $\mathcal{T}$, which is a list of sets of terms
that corresponds to a path in the term tree.
We use operations $\mathtt{push}$, $\mathtt{pop}$, and $\mathtt{last}$, respectively, to add, remove, or read $\mathcal{T}$'s last element.
The algorithm performs a search for each atom in $q$ \lcRef{l_chase_qloop}, adding one current path to $\mathcal{T}$'s root node
after each run \lcRef{l_chase_troot}. The inner loop \lcRef{l_chase_iloop} non-deterministically chooses 
\lcRef{l_chase_choose} a tgd to apply to $\Inter$, or to break the iteration early (even if tgds are applicable).
The current path $\mathcal{T}$ is pruned so that its last element contains a frontier term of the tgd \lcRef{l_chase_prunet},
before we either add a new term set \lcRef{l_chase_add_to_new_bag} or augment the last term set \lcRef{l_chase_add_to_last_bag}, depending on whether
$\arule$ is an $\hat{E}$-tgd.
Finally, we add the inferred head and restrict $\Inter$ to atoms with terms in $\mathcal{T}$ \lcRef{l_chase_update_i}, where $\sigma^+$ extends
$\sigma$ to existentially quantified variables using globally fresh nulls (not used in the algorithm before).
Finally, we check if $\Inter$ entails $q$ \lcRef{l_chase_return}.

Given a run of Algorithm~\ref{algo_pg_chase}, we write $\Inter^j_i$ (resp.\ $\mathcal{T}^j_i$) to denote the value of $\Inter$ (resp.\ $\mathcal{T}$) after executing
\lcRef{l_chase_update_i} in the $i$th iteration of loop \lcRef{l_chase_qloop} and the $j$th iteration of \lcRef{l_chase_iloop}.
Although Algorithm~\ref{algo_pg_chase} can forget inferences and repeat the same tgd application with distinct fresh nulls,
its computations are correct in the following sense:

\begin{restatable}{lemma}{lemmaAlgoPgChaseCorrect}
\label{lemma_algo_pg_chase_correct}
Let $\Inter^*$ be the union of all sets $\Inter^j_i$ of some run of Algorithm~\ref{algo_pg_chase}. 
There is a homomorphism $\tau:\Inter^*\to\fnChase{\aprogram, \Dnter}$, and therefore
$\fnChase{\aprogram, \Dnter}\models q$ whenever Algorithm~\ref{algo_pg_chase} returns \emph{true}. 
\end{restatable}

\begin{restatable}{lemma}{lemmaAlgoPgChaseComplexity}
\label{lemma_algo_pg_chase_complexity}
If $\aprogram$ is arboreous and path-guarded with $\rank{\aprogram}>0$,
and $M\leq f(|\Dnter|)$ for some $\rank{\aprogram}$-exponential function $f$,
then there is a $(\rank{\aprogram}-1)$-exponential function $g$ such that 
Algorithm~\ref{algo_pg_chase} runs in space $g(|\Dnter|)$, where $0$-exponential means polynomial.
\end{restatable}

It remains to show that, whenever $\Sigma,\Dnter\models q$, Algorithm~\ref{algo_pg_chase} admits a run 
that returns true and is bounded by an $M$ as in Lemma~\ref{lemma_algo_pg_chase_complexity}.
Any run corresponds to a sequence of choices for \lcRef{l_chase_choose}, which consists of $|q|$
sequences $\vec{s}_j$ of tgd applications.
Let $\arule$ and $\sigma$ define the tgd application used to compute $\Inter^j_{i+1}$ from $\Inter^j_i$.
We say that this tgd application \emph{corresponds to chase step $s$} if
$\steprule{s}=\arule$,
the restriction of $\tau$ to terms in $\mathcal{T}^j_i$ is injective, and $\sigma(z)=\tau^-(\stepmatch{s}(z))$.
In this case, we \emph{canonically extend} $\tau$ to the fresh nulls 
by $\tau(v\sigma^+)\defeq v\stepmatchex{s}$ for all existential variables $v$ in $\arule$.
This makes $\tau$ locally injective:

\begin{restatable}{lemma}{lemmaAlgoPgChaseCorresponding}
\label{lemma_algo_pg_chase_corresponding}
If all tgd applications in a run of Algorithm~\ref{algo_pg_chase} correspond to chase steps,
and $\tau$ in each iteration is the canonical extension for the respective step,
then $\tau$ is a homomorphism $\Inter^*\to\fnChase{\aprogram, \Dnter}$ that is injective on all term sets $\bigcup\mathcal{T}^j_i$ that occur during the run.
\end{restatable}

For completeness of the algorithm, we are interested in runs where all tgd applications correspond to chase steps.
Such runs can be specified as a list of $|q|$ sequences of chase steps,
where repetitions of steps are allowed (and sometimes necessary).

We obtain a suitable choice sequence by scheduling \emph{tasks} $\tuple{d,i}$, to be read as
``perform chase step $i$ under the assumption that $\Inter$ already contains (isomorphic copies of) all
atoms that can be expressed using terms from the first $d-1$ elements of the path of Lemma~\ref{lemma_local_rule_application} \eqref{item_loc_rule_body}.''
Such tasks may require other tasks to be completed first, since $\Inter$ may not yet contain the whole
premise of $\steprule{i}$:

\begin{definition}\label{def_task_tree}
For chase step $i$, let $\fLineagePath_B(i)$ denote the path of Lemma~\ref{lemma_local_rule_application} \eqref{item_loc_rule_body}.
For an atom $\alpha\in\fnChase{\aprogram, \Dnter}$, let $\fLineagePath(\alpha)$ be the smallest path in the term tree
that contains all terms of $\alpha$ (which are on a path by induction over Lemma~\ref{lemma_local_rule_application} \eqref{item_loc_rule_head}).

The \emph{subtasks} of a task $\tuple{d,i}$ are all tasks $\tuple{e,j}$
with $e\geq d$ a depth, and $j<i$ the largest chase step that produced an atom
$\alpha\in\Dnter^{j+1}\setminus\Dnter^j$ with $|\fLineagePath(\alpha)|=e$ and $\fLineagePath(\alpha)\subseteq\fLineagePath_B(i)$.
The \emph{task tree} for $\tuple{d,i}$ has a root node with label $\tuple{d,i}$ and the task trees for all subtasks of $\tuple{d,i}$
as its children.
Children of a single parent are ordered by the depth in their label: $\tuple{e,j}<\tuple{e',j'}$ if $e<e'$.%
\end{definition}

The atom $\alpha$ in Definition~\ref{def_task_tree} ensures that the application of $\steprule{j}$ in Algorithm~\ref{algo_pg_chase}
will not delete any previous inferences up to depth $e$ (through \lcRef{l_chase_prunet} and \lcRef{l_chase_update_i}).
The order of subtasks ensures that inferences at smaller depths are computed first.
Now the required sequence to successfully perform chase step $i$ in the inner loop of Algorithm~\ref{algo_pg_chase}
is obtained by traversing the task tree for $\tuple{1,i}$ in a topological, order-respecting way (children before parents, smaller sibling nodes before larger ones),
extracting the sequence of chase steps from the second component of the sequence of tasks.

\begin{restatable}{lemma}{lemmaNChaseShortSequences}
\label{lemma_nchase_short_sequences}
If $\aprogram$ is arboreous and path-guarded, $\vec{s}$ is the sequence of chase steps obtained from the task tree with root $\tuple{1,i}$,
then the length $|\vec{s}|$ of $\vec{s}$ is bounded by a $\rank{\aprogram}$-exponential function.
\end{restatable}%
\begin{restatable}{lemma}{lemmaNChaseSingleAtom}
\label{lemma_nchase_single_atom}
If $\aprogram$ is arboreous and path-guarded, $\vec{s}$ is the sequence of chase steps obtained from the task tree for $\tuple{1,i}$,
and $M\geq |\vec{s}|$, then Algorithm~\ref{algo_pg_chase} can choose tgd applications according to $\vec{s}$.
\end{restatable}

Combining these results, we obtain the completeness of our tree-based chase.
Indeed, whenever $q\theta\subseteq \fnChase{\aprogram, \Dnter}$ for some match $\theta$, there are chase steps
$s_1,\ldots,s_{|q|}$ that produce the atoms of $q\theta$.
A suitable strategy then executes Algorithm~\ref{algo_pg_chase} for the choice sequences $\vec{s_i}$
obtained from the task trees for $\tuple{1,s_i}$ with $i=1,\ldots,|q|$.
Lemma~\ref{lemma_nchase_short_sequences} ensures that we can use a suitably small bound $M$ to use the complexity results of
Lemma~\ref{lemma_algo_pg_chase_complexity}, whereas Lemma~\ref{lemma_nchase_single_atom} ensures that the final
atom set $\Inter^{|q|}_{|\vec{s_{|q|}}|}$ contains all atoms of $\tau^-(q\theta)$.

\begin{restatable}{theorem}{theoremExpSpaceCompleteness}
\label{thm_expspace_complete}
	BCQ entailment for arboreous and path-guarded tgd sets $\aprogram$ with $\rank{\aprogram} = \kappa \geq 1$ is $\kExpSpace{(\kappa - 1)}$-complete for data complexity, where $\kExpSpace{0}=\PSpace$.
\end{restatable}

Hardness is shown by reduction from the word problem of $\kappa$-exponentially time-bounded alternating Turing machines.
\PSpace-hardness is illustrated with a simpler reduction from TrueQBF:

\begin{example}
	Let $\phi = \quantify_1 p_1, \ldots, \quantify_\ell p_\ell \ldotp \psi$ be a quantified Boolean formula with
	$\quantify_i \in \{ \exists, \forall \}$ for $1 \leq i \leq \ell$ and
	$\psi = \bigwedge_{j = 1}^k C^j$ in CNF.
	Let $\tilde{p}_i$ denote $p_i$ or its negation $\bar{p}_i$; the clauses $C^i$ are sets of such literals.

	The facts $\Dnter_0 =  \{ \pnEmptySet(\emptyset), \pnLatest(\vdash, \emptyset), \pnNextVar(\vdash, p_1), \pnNextVar(\vdash, \bar{p}_1) \} \cup \{ \pnNextVar(\tilde{p}_{i}, \tilde{p}_{i+1}) \mid 1\leq i<\ell\}$ encode an order over the literals. The tgds $\aprogram_0$ construct literal sets (truth assignments) similar to Example~\ref{ex_sets_saturating}:
	\begin{align} 
		\pnGetSU(x,S) &\to \exists v.\, \pnSU(x, S, v) \wedge \pnSU(x, v, v)\\
		\pnSU(x, S, T) \wedge \pnSU(y, S, S) & \to \pnSU(y, T, T) \\
		\pnLatest(x,S) \wedge \pnNextVar(x,y) & \to \pnGetSU(y,S) \label{eq_qbf_set_trigger}\\
		\pnLatest(x,S) \wedge \pnNextVar(x,y) \wedge \pnSU(y,S,T) &\to \pnLatest(y,T) \label{eq_qbf_set_mark_new}
	\end{align}
	Facts $\pnLatest(l,S)$ mark the latest literal $l$ added to a set $S$.
	The tgd~\eqref{eq_qbf_set_trigger} triggers the creation of a suitable set, and tgd~\eqref{eq_qbf_set_mark_new} marks the newly added literal.
	Note that we never have $\pnSU(p_i,s,s)$ and $\pnSU(\bar{p}_i,s,s)$.
	
	Now $\Dnter_1 = \{ \pnOccurs(\tilde{p}_i,C) \mid \tilde{p}_i \in C^j, 1\leq j\leq k\} \cup
	\{ \pnNext(C^{j-1}, C^j) \mid 1 < j \leq k \} \cup
	\{ \pnFirst(C^1), \pnLast(C^k) \}$ encodes the clauses. We use tgds $\aprogram_1$ to evaluate $\psi$ under a given truth assignment:
	\begin{align}
		\pnSU(x,S,S) \land \pnOccurs(x,c) \land \pnFirst(c) &\to \pnAccumulator(c,S) \label{eq_qbf_conj_first}\\
		\pnAccumulator(c,S) \land \pnNext(c,d) \land \pnSU(x,S,S) \land \pnOccurs(x,d) &\to \pnAccumulator(d,S) \label{eq_qbf_conj_next} \\
		\pnAccumulator(c,S) \land \pnLast(c) &\to \pnTrue(S) \label{eq_qbf_conj_last}
	\end{align}
	Here, $\pnAccumulator(c,s)$ means ``set $s$ satisfies clause $c$''.
	The tgds $\aprogram_1$ then propagate satisfaction along the order of $\Dnter_1$.
	Therefore, $\fnChase{\aprogram_0 \cup \aprogram_1, \Dnter_0 \cup \Dnter_1} \models \exists v \ldotp \pnTrue(v)$ if and only if $\psi$ is satisfiable.
	
	Finally, we use $\Dnter_2 = \{ \pnExVar(p_i) \mid \quantify_i = \exists, 1 \leq i \leq \ell \} \cup \{ \pnPositive(p_i),\allowbreak \pnNegated(\bar{p}_i) \mid 1 \leq i \leq \ell \}$
	and the following tgds $\aprogram_2$ to evaluate $\phi$:
	\begin{align}
		\pnSU(x,S,T) \land \pnExVar(x) \land \pnTrue(T) &\to \pnTrue(S) \label{eq_qbf_eval_ex}\\
		\pnSU(x,S,T) \land \pnPositive(x) \land \pnTrue(T) &\to \pnTruePosSucc(S) \label{eq_qbf_eval_uni_pos} \\
		\pnSU(x,S,T) \land \pnNegated(x) \land \pnTrue(T) &\to \pnTrueNegSucc(S) \label{eq_qbf_eval_uni_neg}\\
		\pnTruePosSucc(S) \land \pnTrueNegSucc(S) &\to \pnTrue(S) \label{eq_qbf_eval_uni_combine}
	\end{align}
	We check satisfaction by evaluating the tree of literal sets by propagating satisfaction from leafs towards the root.
	The tgd~\eqref{eq_qbf_eval_ex} handles existential quantification, and tgds~\eqref{eq_qbf_eval_uni_pos}--\eqref{eq_qbf_eval_uni_combine} handle universal quantification.
	Handling the successors of universal states separately ensures that the tgds are path-guarded.
	Indeed, $\preceq$ of Definition~\ref{def_constraints_new} for the used tgds contains $\ppos{\pnSU}{2} \preceq \ppos{\pnSU}{3}$.
	Note that $\fnChase{\aprogram_0 \cup \aprogram_1 \cup \aprogram_2, \Dnter_0 \cup \Dnter_1 \cup \Dnter_2} \models \exists v \ldotp \pnEmptySet(v) \land \pnTrue(v)$ if and only if $\phi$ is true.
\end{example}

\section{Conclusions and Outlook}\label{sec_conc}

We have established new criteria for chase termination, which advance the state of the art in 
two important ways: (1) they can take advantage of the standard chase,
and (2) they yield new decidable tgd classes with data complexities that
are complete for $\kExpTime{\kappa}$ and $\kExpSpace{\kappa}$, for any $\kappa\geq 0$.
This is obviously too high for transactional DBMS loads, but it allows us to
address a much larger range of complex computational tasks over databases with the chase.
Practical problems of this kind include ontology reasoning \cite{Carral+19:ChasingSets}, database provenance computation \cite{EKM:exruleprov:rr2022},
and querying databases with complex values \cite{MarxKroetzsch:complexValues:ICDT2022}.
We also note that checking our criteria is always dominated by
Datalog reasoning, which is practically feasible and of lower complexity than some established criteria \cite{CG+13:acyclicity}.

Our work brings up many follow-up questions.
First, the new tgd classes are candidates for capturing their respective complexity classes
(at least from $\PSpace$ upwards), but known proof techniques rely on non-saturating
tgds \cite{Bourgaux+:exruleexp:kr21}.
Second, our techniques require a Datalog-first chase strategy, which is avoidable 
for sets and complex values \cite{MarxKroetzsch:complexValues:ICDT2022}.
It is open if similar approaches could apply in our setting.
Third, our criteria can be broadened, e.g., the restriction to a single maximal-rank component in Section~\ref{sec_optchase}
can be relaxed.

Taking a wider view, a central methodological contribution of our work is the labelled dependency graph and its extensive
use for analysing the internal structure of the standard chase. It can be seen as a surrogate for the
more syntactic ``lineage'' of nulls that is available in the semi-oblivious chase -- best exposed through the use of
skolem terms \cite{Marnette09:superWA} --, which has been extremely useful in studying that chase variant.
It is exciting to ask how our method can be similarly useful in further studying the standard chase, e.g., to detect non-termination or to
decide termination in new cases, and whether it can be refined in the style of termination checks based on materialisation
or control flow analysis.

\paragraph*{Acknowledgements} This work was partly supported
by Deutsche Forschungsgemeinschaft (DFG, German Research Foundation) in projects 389792660 (TRR 248, \href{https://www.perspicuous-computing.science/}{Center for Perspicuous Systems}) and 390696704 (CeTI Cluster of Excellence);
by the Bundesministerium für Bildung und Forschung (BMBF) in the \href{https://www.scads.de}{Center for Scalable Data Analytics and Artificial Intelligence} (ScaDS.AI);
by BMBF and DAAD (German Academic Exchange Service) in project 57616814 (\href{https://secai.org/}{SECAI}, \href{https://secai.org/}{School of Embedded and Composite AI});
and by the \href{https://cfaed.tu-dresden.de}{Center for Advancing Electronics Dresden} (cfaed).
 
\bibliographystyle{ACM-Reference-Format}

\clearpage
\appendix
\section{Proofs for Section~\ref{sec_dep}}

\lemmaInfinitePath*
\begin{proof}
If $\fnChase{\aprogram, \Dnter}$ is infinite, it must contain infinitely many nulls.
Let $\prec$ be some total order on nulls of $\fnChase{\aprogram, \Dnter}$ such that
$n\prec m$ holds whenever $n$ was introduced at an earlier chase step than $m$.
Let $G$ be the directed graph that has all such nulls as its vertices and that contains an edge $n\to m$
if $n$ is the $\prec$-largest null with $n\lchaseedgeto{y}m$ in $\fnChase{\aprogram, \Dnter}$ for some $y$.
Then $G$ is acyclic (since $n\lchaseedgeto{y}m$ implies $n\prec m$) and a forest (since each vertex has a unique predecessor by definition).
Root vertices in $G$ correspond to nulls introduced by tgd applications to non-null frontier variables;
since there are only finitely many such root nulls, $G$ (being infinite) contains an infinite tree $T$.
Moreover, $T$ is finitely branching: indeed, $n\to m\in T$ implies that $m$ was produced by a tgd application to frontier terms
that were introduced in $\fnChase{\aprogram, \Dnter}$ no later than $n$; there are only finitely many such terms; hence there are only finitely many such
tgd applications.
As a finitely branching, infinite tree, $T$ must contain an infinite path by Kőnig's lemma, and this path corresponds
to the required chain.
\end{proof}

\section{Proofs for Section~\ref{sec_sat}}

\lemmaPathQuery*
\begin{proof}
The existence of $\mathfrak{p}$ is guaranteed by Lemma~\ref{lemma_ledgraph}.
$\fnApplyPath(\mathfrak{p})\theta\subseteq\Dnter^m$
follows from Definitions~\ref{def_chase} and \ref{def_pathquery}, and the definition of $\theta$.
\end{proof}

\lemmaPropCompl*
\begin{proof}
To check the relevant entailments of the form $\adatalogprogram\models B\to H$ as in \eqref{eq_propbase} and \eqref{eq_propstep}, we can check the entailment $\adatalogprogram, B'\models H'$, where 
$B'$ and $H'$ are sets of atoms obtained by uniformly replacing variables
in $B$ and $H$ with fresh constants.
The claimed complexities are that of Datalog entailment \cite{D+:datalogcomp}.
\end{proof}

\thmSatCompl*
\begin{proof}
Hardness follows from Lemma~\ref{lemma_prop_compl}.
For inclusion, note that Definition~\ref{def_esaturating} \eqref{item_saturating_acyc} can be checked 
in polynomial time for a given $E$. If it holds true, there are at most exponentially many $\bar{E}$-paths
in $C$, leading to exponentially many checks for \eqref{item_saturating_base} and \eqref{item_saturating_step},
which are each in \ExpTime by Lemma~\ref{lemma_prop_compl}.
The last part of the claim follows since there are a polynomial number of strongly connected components,
each with at most exponentially many candidate sets for $E$.
\end{proof}

\lemmaContextsat*
\begin{proof}
Let $e\defeq u\ledgeto{y}v$.
By Lemma~\ref{lemma_ledgraph}, the given chain $\mathfrak{c}$ corresponds to a path $\mathfrak{p}$ in $\ledgraph(\aprogram)$.
The first and last edge of $\mathfrak{p}$ is $e$, so all edges of $\mathfrak{p}$ are in $C$.
We can view $\mathfrak{p}$ as a path of the form
$e_1\mathfrak{c}_1 e_2\cdots e_{k-1}\mathfrak{c}_{k-1}e_k$ with $e_1=e_k=e$; $e_2,\ldots,e_k\in E$;
and $\mathfrak{c}_i$ an $\bar{E}$-path in the sense of Definition~\ref{def_esaturating} for $1\leq i<k$.
For $1\leq i\leq k$, let $n_s^i\lchaseedgeto{x(i)} n_t^i$ be the part of $\mathfrak{c}$ that corresponds to $e_i$,
and let $s(i)$ be the respective chase step (in particular, $s(1)=a$ and $s(k)=b$).

We show by induction that, for all $2\leq i\leq k$, we have
$\Dnter^{s(i)}\models H[n_s^i,\vec{z}\stepmatch{a},n_t^{i-1},\vec{w}\stepmatchex{a}]$.
In particular, for $i=k$ we get 
$\Dnter^b\models H[n_{\ell-1},\vec{z}\stepmatch{a},n_t^{k-1},\vec{w}\stepmatchex{a}]$.
This shows the claim, since it means that Definition~\ref{def_chase} (2.a) would not be satisfied if
$\vec{z}\stepmatch{a}=\vec{z}\stepmatch{b}$.

The base case $i=2$ follows since $e_1\mathfrak{c}_1$ is base-propagating by Definition~\ref{def_esaturating} \eqref{item_saturating_base}.
For the induction step, suppose the claim holds for $i$. 
The claim follows for $i+1$ since $e_i\mathfrak{c}_i e_{i+1}\mathfrak{c}_{i+1}$ is step-propagating for $e$
by Definition~\ref{def_esaturating} \eqref{item_saturating_step}.
\end{proof}

\thmSaturTerm*
\begin{proof}
Suppose for a contradiction that $\fnChase{\aprogram, \Dnter}$ is infinite.
Then there are infinitely many nulls, and a tgd that is applied
infinitely often.
The \emph{tgd of} an edge $\ledge{w'}{w}{x}$ in $\ledgraph(\aprogram)$ is $\arule_w$, the unique tgd with variable $w$.
Now let $C$ be a strongly connected component of $\ledgraph(\aprogram)$ such that 
(1) $C$ contains an edge whose tgd is applied infinitely often, and 
(2) the tgd of every edge $\ledge{w'}{w}{x}$ with $w'\notin C$ and $w\in C$ is applied only finitely often.

Using a similar argument as for the proof of Lemma~\ref{lemma_infinite_path}, we find that there is 
an infinite chain $n_0\lchaseedgeto{y_1} n_1\lchaseedgeto{y_2} \cdots$ in $\fnChase{\aprogram, \Dnter}$
such that $\nullvar{n_i}\in C$ for all $i\geq 0$.
By Lemma~\ref{lemma_ledgraph}, this chain corresponds to an infinite path $\mathfrak{p}$ in $C$.
By assumption, $C$ is $E$-saturating for some set $E$, so, by Definition~\ref{def_esaturating}~\eqref{item_saturating_acyc},
some edge $u \ledgeto{y} v \in E$ occurs infinitely often in $\mathfrak{p}$.
Let $\arule_v$ be its tgd, and let $\rulehead(\arule_v)=\exists v,\vec{w}. H[y,\vec{z},v,\vec{w}]$.

By Definition~\ref{def_esaturating}~\eqref{item_saturating_confluent} and Lemma~\ref{lemma_ledgraph},
applications of $\arule_v$ can only involve null values $n$ for variables of $\vec{z}$
if (i) $\nullvar{n} \notin C$, and (ii) there is an edge from $\nullvar{n}$ to $C$ in $\ledgraph(\aprogram)$.
By our choice of $C$, the number of such nulls $n$ is finite, as is the number of constants
in $\aprogram$ and $\Dnter$, so there are only finitely many possible instantiations of $\vec{z}$
in applications of $\arule_v$.
Together with Lemma~\ref{lemma_contextsat}, this implies that $\arule_v$ is applied only finitely many 
times -- a contradiction.
\end{proof}

\section{Proofs for Section~\ref{sec_rank}}

\lemmaNullsAtDepth*
\begin{proof}
Let $\alpha = \max \{ \card{\vec{y}} \colon B[\vec{x},\vec{y}] \to\exists\vec{v}. H[\vec{y},\vec{v}]\in\aprogram, \vec{v} \cap C \neq \emptyset \}$,
and let $\beta = \confluence{C}$. %
The claim is easy to see for trivial components (case $\beta=0$), 
so we only consider the case $\beta\geq 1$.
We define upper bounds $g(d)$ for the number of nulls of $C$-depth $\leq d$.

Nulls of $C$-depth $d+1$ are created by applying $\arule_v$ with $v \in C$ to a frontier
of at most $\alpha$ terms that are $C$-inputs and nulls of $C$-depth $\leq d$,
where at most $\beta$ of the terms are nulls of $C$-depth $\leq d$.
An upper bound for nulls in $C$ produced by such tgd applications therefore
is $\card{C} i^\alpha g(d)^\beta$.
The total number of nulls at $C$-depth $\leq d+1$ is therefore bounded by
$\card{C} i^\alpha g(d)^\beta + g(d) \leq (\card{C} i^\alpha +1) g(d)^\beta$ (where we use $\beta\geq 1$).
We can therefore define $g(0)\defeq\card{C} i^\alpha +1$ and
$g(d+1)\defeq (\card{C} i^\alpha +1) g(d)^\beta = g(0) g(d)^\beta$ 
to obtain an upper bound for the number of nulls of $C$-depth $\leq d$.

If $\beta=1$, then $g(d)=g(0)^{d+1}$ is exponential in $d$, as claimed.
If $\beta>1$, we use a relaxed upper bound $f(d)\defeq g(0)^{(\beta+1)^d}$.
Indeed, $g(d)\leq f(d)$ follows by induction:
the base case follows from $f(0)=g(0)$; and the step follows from
\begin{align*}
g(d+1) & = g(0) g(d)^\beta \leq g(d) g(d)^\beta = 
g(d)^{\beta+1} \\
  & \leq f(d)^{\beta+1} 
    = \big(g(0)^{(\beta+1)^d}\big)^{\beta+1}
   =      g(0)^{(\beta+1)^{d+1}} 
  = f(d+1)
\end{align*}
where the second ``$\leq$'' uses the induction hypothesis.
Hence, $f$ is the claimed doubly exponential bound.
\end{proof}

\thmBoundedNulls*
\begin{proof}
Let $C_0, \ldots, C_k$ be a topological order as in Definition~\ref{def_rank_comp},
and let $\iota(v)$ denote the number of nulls $n$ with $\nullvar{n}=v$ in $\fnChase{\aprogram, \Dnter}$.
Moreover, let $\kappa$ denote the number of constants in $\Dnter$ and $\aprogram$.
We show the claim for all existential variables $v$ with $\scc{v}=C_i$ by induction over $i=0,\ldots,k$.

Consider $C_i$ and assume that the claim holds for all $C_j$ with $j<i$.
Let $C^-_1, \ldots, C^-_\ell$ with $C^-_j \prec C_i$ be the direct $\prec$-predecessors of $C_i$.
Hence, for all $1\leq j\leq\ell$, there is $m<i$ with $C^-_j=C_m$, and,
for every $v\in C^-_j$, $\iota(v)$ is at most $\rank{C^-_j}$-exponential by induction ($\ast$).
Therefore, the number of $C_i$-inputs is at most $\opfont{in}_i\defeq \kappa + \sum_{j=1}^\ell \sum_{v\in C^-_j}\iota(v)$,
which (by ($\ast$)) is $r_{\opfont{in}}^i$-exponential for $r_{\opfont{in}}^i$ as in Definition~\ref{def_rank_comp}.
Analogously, the number $\opfont{cxt}_i$ of $C_i$-inputs that are nulls $n$ with a $C_i$-incoming edge
$\nullvar{n}\ledgeto{y}v$ such that there is an edge $w\ledgeto{x}v\in E_{C_i}$ with $x\neq y$ is
$r_{\opfont{cxt}}^i$-exponential.

Case $\confluence{C_i} = 0$.
Then $C_i = \{ v \}$ and $\arule_v = B[\vec{x},\vec{y}] \to\exists\vec{v}. H[\vec{y},\vec{v}]$ with $v \in \vec{v}$.
The number of instantiations of $\vec{y}$ and nulls $n$ with $\nullvar{n}=v$ is
bounded by $(\opfont{in}_i)^{\card{\vec{y}}}$. This bound is $r_{\opfont{in}}^i$-exponential and
 $\rank{C_i}=r_{\opfont{in}}^i$, so the claim holds.

Case $\confluence{C_i} \geq 1$.
Let $E$ denote the set $E_{C_i}$.
The $C_i$-depth of nulls has an upper bound that is polynomial in
$\opfont{cxt}_i$.
Indeed, consider an arbitrary chain $\mathfrak{c}$ of nulls $m_j$ with $\nullvar{m_j}\in C_i$ as
in Definition~\ref{def_C_input_depth}.
By Lemma~\ref{lemma_contextsat}, for every $u \ledgeto{y} v \in E$ 
with head $\exists v,\vec{w}. H[y,\vec{z},v,\vec{w}]$ of the corresponding tgd $\arule_v$,
all applications of $\arule_v$ in $\mathfrak{c}$ use a different instantiation of $\vec{z}$.
By Definition~\ref{def_esaturating}~\eqref{item_saturating_confluent}, variables in $\vec{z}$ can only be instantiated
with values from the $\leq\kappa+\opfont{cxt}_i$ many $C_i$-inputs for which $C_i$ has an incoming edge to $v$.
Hence, there are at most $(\kappa+\opfont{cxt}_i)^{|\vec{z}|}$ applications of $\arule_v$ in $\mathfrak{c}$.
Using $\alpha$ to denote the maximal number of frontier variables $\vec{z}$ in any tgd of $E$, there
are at most $|E|\cdot (\kappa+\opfont{cxt}_i)^\alpha$ applications of $E$-tgds in $\mathfrak{c}$.

Now by Definition~\ref{def_esaturating}~\eqref{item_saturating_acyc}, the number of consecutive
tgd applications in $\mathfrak{c}$ that correspond to $\bar{E}$-edges is at most $|\bar{E}|$.
Hence, the overall length of $\mathfrak{c}$ is bounded by
$|\bar{E}|\cdot (|E|\cdot (\kappa+\opfont{cxt}_i)^\alpha+1)$. This bound is polynomial in $\opfont{cxt}_i$, since
$|\bar{E}|$, $|E|$ and $\alpha$ are fixed by $\aprogram$.
Therefore, the $C_i$-depth of nulls is at most polynomial in $\opfont{cxt}_i$, and hence bounded by
an $r_{\opfont{cxt}}^i$-exponential function.

The claim now follows from Lemma~\ref{lemma_nulls_at_depth}, using that
the number of $C_i$-inputs $\opfont{in}_i$ is $r_{\opfont{in}}^i$-exponential as noted above,
where the single and double exponential dependency on the  $C_i$-depth corresponds to the use of
$r_{\opfont{cxt}}^i+1$ and $r_{\opfont{cxt}}^i+2$ in Definition~\ref{def_rank_comp}.
\end{proof}

\thmRankCompl*
\begin{proof}\phantomsection\label{proof_thmRankCompl} %
Theorem~\ref{thm_bounded_nulls} yields a $\rank{\aprogram}$-exponential bound on the number of terms in $\fnChase{\aprogram, \Dnter}$.
The number of atoms in $\fnChase{\aprogram, \Dnter}$, for fixed $\aprogram$, is polynomial in the number of terms.
Since BCQ entailment can be decided over $\fnChase{\aprogram, \Dnter}$, the claimed \kExpTime{\rank{\aprogram}} upper bound follows.

The lower bound can be shown by reduction from the word problem of $k$-exponentially time bounded Turing machines (TMs).
The simulation of TMs with Datalog rules is standard \cite{D+:datalogcomp},
using a strict total order for time steps and tape cells. 
To construct such an order of the required length, we expand the construction of Example~\ref{ex_dexp_saturating} (i.e., the combined tgds from
Examples~\ref{ex_dexp} and \ref{ex_dexp_propagation}) with the following additional tgds:
\begin{align}
	\predname{first}(z)
		\to \predname{min}(\bot,z)\wedge\predname{max}(\top,z) &\wedge\predname{succ}(\bot,\top,z) \label{eq_succ_init}\\
	\predname{cat}(\bar{x}_1,\bar{x}_2,z,x) \wedge \predname{next}(z,z_+) \wedge \predname{up}(x,z_+,\bar{x}) &\to \predname{cnu}(\bar{x}_1,\bar{x}_2,\bar{x},z,z_+)\label{eq_succ_cnu}\\
	\predname{cnu}(\bar{x}_1,\bar{x}_2,\bar{x},z,z_+) \wedge
	\predname{cnu}(\bar{x}_1,\bar{x}'_2,\bar{x}'\!\!,z,z_+) \wedge
	\predname{succ}(\bar{x}_2,\bar{x}'_2,z)
		&\to \predname{succ}(\bar{x},\bar{x}'\!\!,z_+)\label{eq_succ_add_low}\\
	\begin{split}
	\predname{cnu}(\bar{x}_1,\bar{x}_2,\bar{x},z,z_+) \wedge
	\predname{cnu}(\bar{x}'_1,\bar{x}'_2,\bar{x}'\!\!,z,z_+) \wedge{} \\
	\predname{succ}(\bar{x}_1,\bar{x}'_1,z) \wedge
	\predname{max}(\bar{x}_2,z)\wedge\predname{min}(\bar{x}'_2,z)
		&\to \predname{succ}(\bar{x},\bar{x}'\!\!,z_+)
	\end{split}\label{eq_succ_add_high}\\
	\predname{cnu}(\bar{x}_1,\bar{x}_1,\bar{x},z,z_+)\wedge\predname{min}(\bar{x}_1,z)
		&\to \predname{min}(\bar{x},z_+)\label{eq_succ_min}\\
	\predname{cnu}(\bar{x}_1,\bar{x}_1,\bar{x},z,z_+)\wedge\predname{max}(\bar{x}_1,z)
		&\to \predname{max}(\bar{x},z_+)\label{eq_succ_max}
\end{align}
Reading $\bot$-$\top$-sequences as binary numbers, these tgds compute the usual
numeric successor order. In particular, for each $\ell$, facts of the form
$\predname{succ}(\bar{x},\bar{x}'\!,\ell)$, $\predname{min}(\bar{x},\ell)$,
and $\predname{max}(\bar{x},\ell)$ describe a total order of length $2^{2^\ell}$.
The approach is similar to a technique first introduced by \citeauthor{gottlobrr} \cite{gottlobrr},
generalised to make the depth of the construction data-dependent.

We can iterate this construction by using another instance of the above tgds with each predicate
$p$ replaced with a fresh name $p'$, and using tgds like $\predname{last}(z)\wedge \predname{succ}(\bar{x},\bar{x}'\!,z)\to \predname{next}'(\bar{x},\bar{x}')$ to derive an initial order of levels.
The new set of tgds leads to another strongly connected component whose rank is increased by $2$,
since $\predname{next}'$ provides context terms for the renamed version of tgd \eqref{eq_dexp_up}.
This allows us to construct $k$-exponential total orders for all even numbers $k$.

To cover odd ranks as well, we can use the known simulation of nested finite sets 
with tgds \cite{MarxKroetzsch:complexValues:ICDT2022}. 
The construction of single exponentially long chains that has been given for data complexity
in this case \cite[Theorem~6]{MarxKroetzsch:complexValues:ICDT2022}.
Using this construction instead of the above, we obtain exponential chains for
strongly connected components $C$ with $\confluence{C} = 1$.

In summary, we can therefore find, for any number $r\geq 0$, tgds $\aprogram$ with $\rank{\aprogram}=r$
that gives rise to an $r$-exponential chain, and for which the \kExpTime{r}-hard word problem for $k$-exponential time TMs reduces to BCQ entailment.
The special case $r=0$, where $\kExpTime{0}$ denotes $\complclass{PTime}$, agrees with the known data complexity for Datalog \cite{D+:datalogcomp}.
\end{proof}
\section{Proofs for Section~\ref{sec_optchase}}

\lemmaArboreous*
\begin{proof}
We use the notation as in Definition~\ref{def_arboreous}.
By definition, the graph of all edges $t\lchaseedgeto{x} n$ is acyclic, so the null forest $\tuple{\hat{N},\nufoedge}$ is too.
Suppose for contradiction that there is $m\in \hat{N}$ with in-degree $\geq 2$, i.e., 
that there are $n_1,n_2\in\hat{N}$ with $n_i\lchaseedgeto{y_i} m$ for $i\in\{1,2\}$.
Then $y_1 \neq y_2$, since different nulls $n_1,n_2$ must match different variables
to be used in one tgd application.
By definition, $\nullvar{n_1},\nullvar{n_2}\in\hat{C}$, and therefore $\confluence{\hat{C}}\geq2$.
A contradiction.
\end{proof}

\lemmaFacNullForest*
\begin{proof}
We use the notation as in Definition~\ref{def_fac_null_forest}.
By definition, the sets $R[i]$ are disjoint, since each null is created by a unique tgd application.
The sets $F[i]$ are disjoint from all $R[j]$ with $j\neq i$ by construction.
Since $\tuple{\hat{N},\nufoedge}$ is a tree, the sets $F[i]$ are therefore mutually disjoint.
Since $\locdec_0$ contains all remaining terms by construction, $N_\sim$ is indeed a partition of the terms of the chase.

The relation $\fanufoedge$ is acyclic on $\mathcal{F}$. Indeed, suppose that there were a $\fanufoedge$-cycle $C$ in 
$\mathcal{F}$. Each $F[i]\fanufoedge F[j]$ in $C$ corresponds to a relation $n_1\nufoedge n_2$ for nulls $n_1,n_2$
as in Definition~\ref{def_fac_null_forest}. Since $F[i]\neq F[j]$, $n_2\in R[j]$.
Stemming from a single tgd application, all elements in $R[j]$ have the same predecessors, hence 
there is a $\nufoedge$-path from $n_1$ to every $n\in F[j]$. Since we obtain such paths for every
edge in the $\fanufoedge$-cycle $C$, we find a $\nufoedge$-cycle in the null forest. This contradicts the
acyclicity of the null forest (Lemma~\ref{lemma_arboreous}).

With the additional root element $\locdec_0$, $\mathcal{F}$ therefore becomes a tree as required.
\end{proof}

We make another small observation that was not included in the main text of the paper, but
that is relevant to avoid Definition~\ref{def_fac_null_forest} from requiring further special cases.

\begin{applemma}\label{applemma_erule_to_scc}
Let $\aprogram$ be arboreous with related notation as in Definition~\ref{def_fac_null_forest}.
Then $V_{\hat{E}}\subseteq\hat{C}$, i.e., all existential variables of all $\hat{E}$-tgds
are contained in $\hat{C}$.
\end{applemma}
\begin{proof}
Consider some edge $w\ledgeto{y} v$ in $\hat{E}$. Then $v\in\hat{C}$.
Let $\vec{v}$ be the set of existential variables in the tgd that contains $v$.
Then every $v'\in\vec{v}$ also satisfies $v'\in\hat{C}$, since otherwise $v'$ would be
part of a distinct strongly connected component that would have the same or a greater rank than $\hat{C}$,
contradicting the requirement that $\hat{C}$ is the unique SCC of maximal rank (Definition~\ref{def_arboreous}).
\end{proof}

The previous result clarifies possible uncertainty about the sets $R[i]$ in Definition~\ref{def_fac_null_forest}:
even when applied to a match that only includes terms that are not in the null forest $\hat{N}$, the resulting
set of fresh nulls is fully contained in $\hat{N}$.

\lemmaConstraintsNew*
\begin{proof}
Let $\Dnter^0, \Dnter^1, \ldots$ be the chase sequence of $\fnChase{\aprogram, \Dnter}$.
We show the claim for $p(\vec{t})\in\Dnter^c$ by strong induction on $c>0$.
For $\Dnter^0 = \Dnter$, $t_i\notin\hat{N}$ so $t_i\in \locdec_0$, and the claim holds since $\locdec_0$ is
the root of the term tree.

For the induction step $\Dnter^{c+1}$, we only consider the case $t_i\in\hat{N}$ (the case $t_i\notin\hat{N}$ works as before),
and we show that $\locdec(t_i)\fanufoedgeStar\locdec(t_j)$.
Note that this implies $t_j\in\hat{N}$, since there is no edge from $\locdec(t_i)$ to $\locdec_0$.

Therefore, let $\arule\defeq\steprule{c} = B[\vec{x},\vec{y}] \to\exists\vec{v}.H[\vec{y},\vec{v}]$ and let $\sigma^+ \defeq \stepmatchex{c}$.
Moreover, assume that $\ppos{p}{i}\preceq\ppos{p}{j}$ and $p(t_1,\ldots,t_k) = p(z_1\sigma^+,\ldots,z_k\sigma^+) \in H\sigma^+$
with $t_i\in\hat{N}$. We have $z_i,z_j\in\vec{y}\cup\vec{v}$. We distinguish the possible cases:

\begin{enumerate}
\item Case $z_i, z_j\in\vec{v}$. Since $t_i\in\hat{N}$, $z_i\in\hat{C}$. By Definition~\ref{def_constraints_new} \eqref{item_constraints_root}, $z_j\in\hat{C}$
and therefore $t_j\in\hat{N}$. By Definition~\ref{def_fac_null_forest}, $t_i,t_j\in R[c]$, so $\locdec(t_i)=\locdec(t_j)$ and $\locdec(t_i)\fanufoedgeStar\locdec(t_j)$.
\item Case $z_i,z_j\in\vec{y}$. Since $\ppos{p}{i}\preceq\ppos{p}{j}$, Definition~\ref{def_constraints_new} \eqref{item_constraints_propagation} implies
that there is a chain of body variables $z_i=x_1\trianglelefteq\ldots\trianglelefteq x_\ell=z_j$ in $\arule$. By definition of $\trianglelefteq$,
each pair of adjacent variables $x_m,x_{m+1}$ in that chain occurs at $\preceq$-comparable body positions $\ppos{q}{i_m}\preceq\ppos{q}{i_{m+1}}$ in a
body atom $q(\vec{s})$ such that $q(\vec{s})\sigma^+ \in \bigcup_{d\leq c}\Dnter^d$. By our induction hypothesis, the claim holds for $x_m\sigma^+$ and $x_{m+1}\sigma^+$.
Since $\fanufoedgeStar$ is transitive, this shows the claim for $t_i$ and $t_j=x_\ell\sigma^+$.
\item Case $z_i\in\vec{y}$ and $z_j\in\vec{v}$. Then $t_i\lchaseedgeto{z_i}t_j$ in $\fnChase{\aprogram, \Dnter}$
and $\nullvar{t_i}\ledgeto{z_i}z_j$ in $\ledgraph(\aprogram)$ by Lemma~\ref{lemma_ledgraph}.
$\ledgraph(\aprogram)$ has a unique rank-maximal strongly connected component $\hat{C}$ by Definition~\ref{def_fac_null_forest},
and $\nullvar{t_i}\in\hat{C}$ by $t_i\in\hat{N}$. Therefore $\nullvar{t_i}\ledgeto{z_i}z_j$ implies $z_j=\nullvar{t_j}\in\hat{C}$,
hence $t_j\in\hat{N}$.
By Definition~\ref{def_fac_null_forest}, $\locdec(t_i)=\locdec(t_j)$ or $\locdec(t_i)\fanufoedge\locdec(t_j)$, so $\locdec(t_i)\fanufoedgeStar\locdec(t_j)$.
\item Case $z_i\in\vec{v}$ and $z_j\in\vec{y}$. By Definition~\ref{def_constraints_new} \eqref{item_constraints_newbag}, $z_i\notin V_{\hat{E}}$,
and by Definition~\ref{def_constraints_new} \eqref{item_constraints_newstart} $z_j\in\lambda(z_i)$.
The latter implies that $\nullvar{t_j}\in\hat{C}$, and therefore $t_j\in\hat{N}$.
Since $z_i\notin V_{\hat{E}}$, $\arule$ is not an $\hat{E}$-tgd as in Definition~\ref{def_fac_null_forest}, and therefore $t_i\in\locdec(t_j)$,
i.e., $\locdec(t_i)=\locdec(t_j)$ and $\locdec(t_i)\fanufoedgeStar\locdec(t_j)$. \qedhere
\end{enumerate}
\end{proof}

\lemmaLocalRuleApplication*
\begin{proof}
Item \eqref{item_loc_rule_body}.
Since $\steprule{i}$ is path-guarded, the $\hat{C}$-affected variables $\vec{x}'\cup\vec{y}'\subseteq\vec{x}\cup\vec{y}$ form
a chain in $\trianglelefteq$.
By definition of $\trianglelefteq$, we can order $\vec{x}'\cup\vec{y}'$ as a sequence $z_1,\ldots,z_\ell$
such that, for every $j\in\{1,\ldots,\ell-1\}$, there is a body atom $p(\vec{s})$ in $\steprule{i}$, 
such that $z_j$ and $z_{j+1}$ occur at positions $\ppos{p}{a}\preceq\ppos{p}{b}$.

Applying Lemma~\ref{lemma_constraints_new} inductively to the sequence $z_1,\ldots,z_\ell$, we find
that, for all $j\in\{1,\ell-1\}$, $\locdec(z_{j})\fanufoedgeStar\locdec(z_{j+1})$.
For remaining variables $z\in\vec{x}\cup\vec{y}$ that are not $\hat{C}$-affected, we have that $z\stepmatch{i}\notin\hat{N}$
by Lemma~\ref{lemma_ledgraph} and the definition of $\hat{N}$; hence $\locdec(z)=\locdec_0$ in these cases.
In summary, if $T_B\neq\emptyset$, then $T_B\subseteq\bigcup\fLineagePath(z_\ell\stepmatch{i})$.

Item \eqref{item_loc_rule_head}.
The claim is obvious for tgds with $\vec{v}=\emptyset$, since it follows from item \eqref{item_loc_rule_body} in this case.
For $\vec{v}\neq\emptyset$ and $\vec{v}\cap\hat{C}=\emptyset$, the claim is also obvious since in this case
$T_H\subseteq\locdec_0$.
Next, consider $\vec{v}\cap\hat{C}\neq\emptyset$.
Since $\confluence{\hat{C}}=1$, there is at most one $y\in\vec{y}$ such that $y\stepmatch{i}\in\hat{N}$.
If there is no such $y$, then $\vec{y}\stepmatch{i}\subseteq\locdec_0$.
\begin{itemize}
\item Case 1: $\steprule{i}$ is an $\hat{E}$-tgd. Then $\vec{v}\cap\hat{C}\neq\emptyset$ implies
$\vec{v}\cap V_{\hat{E}}\neq\emptyset$, and therefore $\vec{v}\subseteq V_{\hat{E}}$ by Lemma~\ref{applemma_erule_to_scc}.
Hence, $R[i]=\vec{v}\stepmatchex{i}$ and $F[i]$ has no predecessors in $\mathcal{F}$, so $\locdec_0\fanufoedge F[i]$.
This shows that $T_H\subseteq\bigcup\fLineagePath(n)$ for any $n\in\vec{v}\stepmatchex{i}$.
\item Case 2: $\steprule{i}$ is not an $\hat{E}$-tgd. Then $T_H\subseteq\locdec_0$, so the claim holds too.
\end{itemize}
Alternatively, assume that there is $y$ with $y\stepmatch{i}\in\hat{N}$.
\begin{itemize}
\item Case 1: $\steprule{i}$ is an $\hat{E}$-tgd. Then $R[i]=\vec{v}\stepmatchex{i}$ follows as in Case~1 above.
Since $y\stepmatch{i} \nufoedge v\stepmatchex{i}$ holds for any $v\in\vec{v}$, we have
$\locdec(y\stepmatch{i})\fanufoedge\locdec(v\stepmatchex{i})=F[i]$. We obtain $T_H\subseteq\bigcup\fLineagePath(v\stepmatchex{i})$
since all $y'\in\vec{y}$ with $y'\neq y$ are such that $y'\stepmatch{i}\in\locdec_0$ (since $\confluence{\hat{C}}=1$).
\item Case 2: $\steprule{i}$ is not an $\hat{E}$-tgd. Then $\locdec(y\stepmatchex{i})=\locdec(v\stepmatchex{i})$
for every $v\in\vec{v}$, and hence $T_H\subseteq\bigcup\fLineagePath(y\stepmatchex{i})$ using the same argument as in the previous case for $y'\neq y$.\qedhere
\end{itemize}
\end{proof}

\lemmaAlgoPgChaseCorrect*
\begin{proof}
We iteratively define $\tau$ for fresh nulls introduced during a run of Algorithm~\ref{algo_pg_chase},
and verify the claimed homomorphism property for each step.
Since the outer loop does not matter here, we use $\Inter_0,\ldots,\Inter_\ell$ to denote the entire sequence of
values for $\Inter$ as they occur throughout the algorithm.

Initially, $\tau$ is the identity function on constants in $\aprogram$ and $\Dnter$, which is a
homomorphism from the initially empty set $\Inter_0$ to $\fnChase{\aprogram, \Dnter}$.

Now by way of induction, assume that $\tau$ has been defined so that it is a homomorphism $\bigcup_{i=0}^n\Inter_i\to\fnChase{\aprogram, \Dnter}$
for some $n\geq 0$, and that a further tgd application with tgd $\arule$ and match $\sigma$ is chosen in \lcRef{l_chase_choose}.
Since $\sigma$ is a match for $\arule$ on $\Inter_n$, we find a corresponding match
$\sigma_\bot$ of $\arule$ on $\fnChase{\aprogram, \Dnter}$ where $\sigma_\bot(z)=\tau(\sigma(z))$ for all variables $z$ in the body of $\arule$.
Since this match $\sigma_\bot$ is satisfied in $\fnChase{\aprogram, \Dnter}$, it can be extended to a match $\sigma_\bot^+$ such that 
$\rulehead(\arule)\sigma_\bot^+\subseteq\fnChase{\aprogram, \Dnter}$. 
Therefore, given the extended match $\sigma^+$ that is used in Algorithm~\ref{algo_pg_chase} to apply $\arule$,
we define $\tau(v\sigma^+)$ for all existential variables $v$ in $\arule$ as $\tau(v\sigma^+)=v\sigma_\bot^+$.
Then $\tau$ is a homomorphism $\bigcup_{i=0}^{n+1}\Inter_i\to\fnChase{\aprogram, \Dnter}$ as required.
\end{proof}

\lemmaAlgoPgChaseComplexity*
\begin{proof}
Using binary encoding, the numbers $i\leq M \leq f(|\Dnter|)$ can be stored in $(\rank{\aprogram}-1)$-exponential space.
To show that $\Inter$ can be stored in $(\rank{\aprogram}-1)$-exponential space, note that the sets of $\mathcal{T}$,
other than the root, correspond to nodes $\locdec(n)$ of the term tree in the following sense:
if $\mathcal{T}[d]$ ($d\in\{1,\ldots,|\mathcal{T}|\}$) is the $(d+1)$-th element in $\mathcal{T}$ (the first $\mathcal{T}[0]$ being the root),
then there is a node $\locdec\in N_\sim$ that is $d$ steps away from the root such that
$\tau(\mathcal{T}[d])\subseteq\locdec$, with $\tau$ as in Lemma~\ref{lemma_algo_pg_chase_correct}.

The terms in $\bigcup_{d=1}^{|\mathcal{T}|}\mathcal{T}[d]$ are therefore always contained in a single path of the term tree.
The length of paths of the null forest (and analogously in the term tree) are bounded by a $(\rank{\aprogram}-1)$-exponential function,
as shown in the proof of Theorem~\ref{thm_bounded_nulls}. Indeed, that proof establishes the $C_i$-depth of nulls in a strongly connected
component $C_i$ is exponentially bounded in $r_{\opfont{cxt}}^i$ and polynomially bounded in $r_{\opfont{in}}^i$.
For $\rank{C_i}>0$, the latter corresponds to a $(r_{\opfont{in}}^i - 1)$-exponential bound.
The claim about $\hat{C}$ follows since path lengths in the null forest corresponds to the $\hat{C}$-depth of nulls,
and since $\rank{\aprogram}=\rank{\hat{C}}\geq\max\{r_{\opfont{in}},r_{\opfont{cxt}}+1\}$ by Definition~\ref{def_rank_comp}.

The size of the sets $\locdec(n)$ is polynomial in $|\Dnter|$, since $\locdec(n)$ only contains nulls from
applications of non-$\hat{E}$-tgds, which do not have a dependency cycle, so that the polynomial data complexity of
jointly acyclic tgds applies \cite{KR11:jointacyc}. Therefore, the size of $\bigcup_{d=1}^{|\mathcal{T}|}\mathcal{T}[d]$ is
$(\rank{\aprogram}-1)$-exponentially bounded in $|\Dnter|$.

This bound also applies to the initial value of $\mathcal{T}$ from \lcRef{l_chase_tinit}, so that after $|q|$ executions of loop
\lcRef{l_chase_qloop}, the bound remains $(\rank{\aprogram}-1)$-exponential (note that $|q|$ is constant with respect to $\Dnter$). 

With the overall set of available terms restricted by a $(\rank{\aprogram}-1)$-exponential bound in $|\Dnter|$, this bound carries over 
to the possible atoms in $\Inter$ throughout the computation. The final check in \lcRef{l_chase_return} can also be performed
in this space bound, e.g., by iterating over all possible variable bindings with respect to $\mathcal{T}$.
\end{proof}

\lemmaAlgoPgChaseCorresponding*
\begin{proof}
Note that our requirements for ``corresponding to a chase step'' already include the injectivity of $\tau$
on the terms used in the premise. Nevertheless, the claim is still non-trivial, since the final iteration of each run
of the inner loop in Algorithm~\ref{algo_pg_chase} is not covered by the requirements, and since the algorithm, by virtue of being able to
non-deterministically break the computation at any time, can certainly perform runs that satisfy the preconditions.
In particular, the required injectivity holds for the initial term set $C_0$ for which $\tau$ was defined as the identity.

We proceed by induction. Consider the tgd application that produces
$\Inter^{j+1}_{i}$ from $\Inter^j_i$. By assumption, it corresponds to a chase step $s$.
Let $\vec{v}$ be the set of existential variables in $\steprule{s}$.

Now suppose for a contradiction that the canonical extension of $\tau$ in this step is not injective.
By our definition, $\tau$ induces a bijection $\vec{v}\sigma^+ \to \vec{v}\stepmatchex{s}$, and 
it is injective on $\mathcal{T}^j_i$ (a precondition for the application corresponding to step $s$).
Hence, the supposed violation of injectivity requires that there is a null $n\in\vec{v}\sigma^+$
and a term $t\in\bigcup\mathcal{T}^j_i$ such that $\tau(n)=\tau(t)$. Since $\tau(n)\in\vec{v}\stepmatchex{s}$,
$t$ must also be a null (constants are always mapped to themselves in $\tau$).
By the assumption, $\tau$ has been defined through a series of canonical extensions, so
the value $\tau(t)$ was assigned in a previous tgd application that also corresponded to step $s$
(since $\tau(t)\in\vec{v}\stepmatchex{s}$ can be a fresh null only for this one step).
Let $\theta^+$ be the extended match used in this tgd application (it has to agree with $\sigma$ on 
universal variables, but must use different nulls), hence $t\in\vec{v}\theta^+$.
But then $t$ was added to $\mathcal{T}$ in \lcRef{l_chase_add_to_new_bag} or \lcRef{l_chase_add_to_last_bag}.
In either case, the whole set $\vec{v}\theta^+$ occurs in the same set of $\mathcal{T}$ that also contains $t$,
so that $\vec{v}\theta^+\subseteq\bigcup\mathcal{T}^j_i$. In this case, however, $\steprule{s}\theta^+\subseteq\Inter^j_i$,
so $\steprule{s}$ is not applicable to obtain $\Inter^{j+1}_{i}$. A contradiction.
\end{proof}

\lemmaNChaseShortSequences*
\begin{proof}
Let the \emph{knobbly term tree} be obtained from the term tree by simultaneously replacing each node set $\locdec\in\ N_\sim$ with
the union $\locdec \cup\bigcup_{\locdec\fanufoedge\locdec_c} \locdec_c$ that also includes all terms in the node's direct children.
The tree structure otherwise remains the same, i.e., the term tree and the knobbly term tree are isomorphic.
In particular, the length of paths in the knobbly term tree is bounded by a $(\rank{\aprogram}-1)$-exponential function, as observed
for the term tree in the proof of Lemma~\ref{lemma_algo_pg_chase_complexity}.

Now consider any path $\tuple{d_1,s_1}\cdots\tuple{d_\ell,s_\ell}$ in the task tree, and let 
$\alpha_i$ denote the atoms $\alpha$ of Definition~\ref{def_task_tree} for every subtask $\tuple{d_i,s_i}$
with $1<i\leq\ell$.

For a path $\mathfrak{p}$, let $\mathfrak{p}|_d$ denote the path of the initial $d$ nodes in $\mathfrak{p}$.
We claim that for all $i\in\{2,\ldots,\ell\}$, $\fLineagePath(\alpha_i)|_{d_i-1}=\fLineagePath(\alpha_\ell)|_{d_i-1}$,
i.e., the paths of atoms $\alpha_i$ agree with the path of $\alpha_\ell$, except possibly for the lowest node (Claim~$\ddagger$).
This is trivial for $i=\ell$.
For a task $\tuple{d_i,s_i}$ with $1<i<\ell$, assume by way of induction that the claim was shown for $\tuple{d_{i+1},s_{i+1}}$.
Then $\fLineagePath(\alpha_{i+1})|_{d_{i+1}-1}=\fLineagePath(\alpha_\ell)|_{d_{i+1}-1}$ by induction hypothesis,
and $\fLineagePath(\alpha_{i+1})|_{d_{i+1}-1}\subseteq\fLineagePath(\alpha_{i+1})\subseteq\fLineagePath_B(s_i)$ by Definition~\ref{def_task_tree}.
Thus, $\fLineagePath_B(s_i)|_{d_{i+1}-1}\subseteq\fLineagePath(\alpha_\ell)|_{d_{i+1}-1}$, and, as $d_{i+1}\geq d_i$,
$\fLineagePath_B(s_i)|_{d_i-1}\subseteq\fLineagePath(\alpha_\ell)|_{d_i-1}$ ($\ast$).
Since $|\fLineagePath(\alpha_i)|=d_i$, there are two cases:
\begin{itemize}
\item $\fLineagePath(\alpha_i)=\fLineagePath_B(s_i)|_{d_i}$ (then $\alpha_i$ contains no new nulls, or $\steprule{s_i}$ is not an $\hat{E}$-tgd).
\item $\fLineagePath(\alpha_i)$ extends the path $\fLineagePath_B(s_i)|_{d_i-1}$ by one additional child node (then $\alpha_i$ contains new nulls from the $\hat{E}$-tgd $\steprule{s_i}$).
\end{itemize}
In either case, $\fLineagePath(\alpha_i)|_{d_i-1}\subseteq\fLineagePath_B(s_i)|_{d_i-1}\subseteq\fLineagePath(\alpha_\ell)|_{d_i-1}$ as required,
where the second $\subseteq$ is ($\ast$).

Now let $\fLineagePath_\bullet(\alpha_\ell)$ denote the path in the knobbly term tree that corresponds to $\fLineagePath(\alpha_\ell)$
by the isomorphism. By Claim~$\ddagger$, for every $i\in\{2,\ldots,\ell\}$, $\fLineagePath(\alpha_i)\subseteq\fLineagePath_\bullet(\alpha_\ell)$.
Containment is clear for $\fLineagePath(\alpha_i)|_{d_i-1}$ by Claim~$\ddagger$. Since $|\fLineagePath(\alpha_i)|=d_i$, the path has at most
one additional final node not in $\fLineagePath(\alpha_\ell)$, and the terms of this node, being a direct child, are included in $\fLineagePath_\bullet(\alpha_\ell)$.

Since $|\fLineagePath_\bullet(\alpha_\ell)|$ is bounded by a $(\rank{\aprogram}-1)$-exponential function, and since the term sets that constitute
the nodes are still of constant size (being unions of terms generated by jointly-acyclic sets of tgds, cf.\ proof of Lemma~\ref{lemma_algo_pg_chase_complexity}),
the cardinality of $\bigcup\fLineagePath_\bullet(\alpha_\ell)$ is also bounded by a $(\rank{\aprogram}-1)$-exponential function.
But then, given the fixed signature of $\aprogram$, there are at most $(\rank{\aprogram}-1)$-exponentially many atoms that can play the
role of $\alpha_i$ ($2\leq i\leq\ell$) in the above path, and since each atom is produced in just one chase step, the path corresponds to a
(strictly decreasing) sequence of at most $(\rank{\aprogram}-1)$-exponentially many chase steps.

This shows that the depth of the task tree is bounded by a $(\rank{\aprogram}-1)$-exponential function, so the size of
the task tree (and of the induced sequence of steps) is bounded by a $\rank{\aprogram}$-exponential function.
\end{proof}

\lemmaNChaseSingleAtom*
\begin{proof}
The claim refers to a single execution of the inner loop of Algorithm~\ref{algo_pg_chase}, starting from $\Inter_0=\Dnter$.
Let $\ell=|\vec{s}|$ be the length of $\vec{s}$, let $\Inter_i$ for $1\leq i\leq\ell$ denote the value of $\Inter$ after 
$i$ iterations, and let $\vec{s}[i]$ be the $i$th element of $\vec{s}$.
Moreover, let $\opfont{task}[i]$ be the task with label $\tuple{d,\vec{s}[i]}$
that gave rise to $\vec{s}[i]$ in $\vec{s}$, and let $\opfont{depth}[i]=d$ be its depth. 
If index $i$ corresponds to a subtask in the term tree, let $\alpha[i]$ be
the atom $\alpha$ of Definition~\ref{def_task_tree} that justified its inclusion as a child node.
As before, given a path $\mathfrak{p}$, we write $\mathfrak{p}|_d$ for the path of the initial $d$ nodes in $\mathfrak{p}$.

We show by induction over $i\in\{1,\ldots,\ell\}$ that for all atoms $\alpha\in\fnChase{\aprogram, \Dnter}$
with $\alpha\in\Dnter^{\vec{s}[i]-1}$ and $\fLineagePath(\alpha)\subseteq\fLineagePath_B(\vec{s}[i])$, there is a corresponding
atom $\beta=\tau^-(\alpha)\in\Inter_{i-1}$, where $\tau^-$ is well-defined by Lemma~\ref{lemma_algo_pg_chase_corresponding}.
In particular, this shows that
\begin{enumerate}
\item the match $\stepmatch{\vec{s}[i]}$ of $\steprule{\vec{s}[i]}$ has a corresponding match on $\Inter_{i-1}$ via $\tau^-$,
\item if $\steprule{\vec{s}[i]}$ contains an existential variable, then no Datalog rule in $\aprogram$ is applicable to $\Inter_{i-1}$, and
\item Algorithm~\ref{algo_pg_chase} can execute all non-redundant tgd applications in the given sequence, and $\Inter_i$ contains
an instance of the head of $\steprule{\vec{s}[i]}$.
\end{enumerate}
Item (1) is clear. Item (2) follows since the sequence of chase steps in $\fnChase{\aprogram, \Dnter}$ also respects the Datalog-first condition,
and since the conclusions of Datalog rules over $\Inter_{i-1}$ can only use the terms in $\Inter_{i-1}$, and in particular are in
$\fLineagePath(\beta)\subseteq\fLineagePath_B(\vec{s}[i])$. Together, (1) and (2) ensure that $\steprule{\vec{s}[i]}$ is applicable
at step $i$ of Algorithm~\ref{algo_pg_chase} if its head is not already satisfied in $\Inter_{i-1}$.
Note that the latter case can only occur if the same tgd application has been performed before, since earlier
chase steps $<\vec{s}[i]$ have not prevented the application of $\steprule{\vec{s}[i]}$ in $\fnChase{\aprogram, \Dnter}$.
In this situation, we ignore step $i$ and continue immediately with the next choice $i+1$ (if any), and we let $\Inter_i=\Inter_{i-1}$.
Hence we also obtain (3).

Now let $i\in\{1,\ldots,\ell\}$ and assume that the induction claim holds true for all $i'<i$.
Consider an arbitrary atom $\beta$ as in the claim.
Let $d_\beta\defeq|\fLineagePath(\beta)|$ be the depth of $\beta$, and let
$s_\beta$ be the chase step that produced $\beta\in\fnChase{\aprogram, \Dnter}$. We claim that $\beta\in\Inter_{i-1}$.

Case (i). If $d_\beta<\opfont{depth}[i]$, then $\opfont{depth}[i]>1$.
Let $k$ be the ancestor node of $i$ that is closest to $i$ (i.e., lowest in the task tree),
such that $\opfont{depth}[k]\leq d_\beta$.
By Definition~\ref{def_task_tree}, $\vec{s}[k]>\vec{s}[i]>s_\beta$, so
$k$ has a child node $j$ with label $\tuple{d_\beta,a}$ where $a\geq s_\beta$.
Then $\beta\in\Inter_{j}$ by the induction hypothesis. Moreover, due to the traversal order
of children of $k$, all nodes between position $j$ and $i$ have depth $>d_\beta$.
This ensures that $\beta\in\Inter_{i-1}$ as required (we give a more detailed account of this argument for a slightly more general
situation in Case (ii)).

Case (ii). If $d_\beta\geq\opfont{depth}[i]$, then $\opfont{task}[i]$ has a descendant node $j$ the task tree with label
$\tuple{d_\beta,s_\beta}$.
This is easy to see for $d_\beta=1$, since the path of depth $1$ is unique, so that
the condition $\fLineagePath(\alpha)\subseteq\fLineagePath_B(i)$ in Definition~\ref{def_task_tree} is tautological if
$|\fLineagePath(\alpha)|=1$. Hence the chase steps for atoms at depth $1$ appear within a single path below $i$ (with chase steps 
of such atoms in decreasing order, largest first).

For $d_\beta>1$, we find a similar path below task $i$. Care is needed since the condition
in Definition~\ref{def_task_tree} refers to the body path of the immediate parent node, which may not be
the body path of $i$, since only the nodes up to depth $d_\beta-1$ are stable yet.
We therefore make the following observation: The earliest chase step $c$ that produces an atom $\beta$ such that 
$|\fLineagePath(\beta)|=d_\beta$ and $\fLineagePath(\beta)\subseteq\fLineagePath_B(i)$ must be the step that
introduced the set of nulls denoted $R[c]$ in Definition~\ref{def_fac_null_forest}, i.e., that initialised the
node $F[i]$ of $\fLineagePath_B(i)$ at depth $d_\beta$. All other chase steps $a>c$ that infer an atom $\alpha$ with
$|\fLineagePath(\alpha)|=d_\beta$ and $\fLineagePath(\alpha)\subseteq\fLineagePath_B(i)$ have a frontier
variable that is matched to a term in $F[i]$. Therefore, $F[i]$ is a node in $\fLineagePath_B(a)$, and every
atom $\gamma$ with $\fLineagePath(\gamma)\subseteq\fLineagePath_B(i)$ also satisfies $\fLineagePath(\gamma)\subseteq\fLineagePath_B(a)$.
The chase steps that produce such atoms therefore form a sequence $c<a_1<\ldots<a_m$, and we find an according path of tasks
$\tuple{d_\beta,a_m}\to\cdots\to\tuple{d_\beta,a_1}\to\tuple{d_\beta,c}$ in the task tree.
This finishes the argument that we find the claimed descendant node $j$ of $i$.

Then $j<i$ since the task tree is traversed in topological order.
Let $\arule_\beta\defeq\steprule{s_\beta}$, $\sigma_\beta\defeq\stepmatch{s_\beta}$, and $\sigma^+_\beta\defeq\stepmatchex{s_\beta}$.
By the induction hypothesis, the tgd application for this step $s_\beta=\vec{s}[j]$ succeeded
with a match $\theta=\tau^-(\sigma_\beta)$%
\footnote{We write $\tau^-(\sigma_\beta)$ for the function that maps $z$ to $\tau^-(z\sigma_\beta)$.
This is sometimes denoted as $\tau^-\circ\sigma_\beta$ but sometimes also as $\sigma_\beta\circ\tau^-$; we avoid this confusion.},
and was performed with an extended match $\theta^+=\tau^-(\sigma^+_\beta)$.
However, it is possible that the deletions in $\mathcal{T}$ between step $j$ and step $i$ were such that
$\tau^-$ (which is only defined locally) is not the same at both steps, hence we cannot yet conclude
$\beta\in\Inter_j$ but merely that $\beta'\in\Inter_j$ for some variant of $\beta$ that might use different fresh nulls.

We therefore show by induction that all intermediate steps $k$ with $j<k<i$ are such that, after Algorithm~\ref{algo_pg_chase}
has executed \lcRef{l_chase_prunet}, $\mathcal{T}$ has length $\geq d_\beta$. This shows that any fresh nulls of step $j$
are still in $\mathcal{T}$ at step $i$, and this implies $\beta'=\beta\in\Inter_{i+1}$ as required.
Since $j$ is part of the path $\tuple{d_\beta,a_m}\to\cdots\to\tuple{d_\beta,a_1}\to\tuple{d_\beta,c}$ as defined above,
there are two options for nodes $k$:
(i) $\opfont{task}[k]=\tuple{d_\beta,a}$ for some $a\in\{a_1,\ldots,a_m\}$, or 
(ii) $\opfont{task}[k]=\tuple{e,a}$ for some $e>d_\beta$.
In case (i), as noted above, $\steprule{a}$ has a frontier variable that is matched to a term in $F[i]$,
which shows the claim about $k$ since $F[i]$ is the element at position $d_\beta$ in $\mathcal{T}$ by induction hypothesis.
In case (ii), the claim likewise follows since $\steprule{a}$, in order for $\alpha[k]$ to be at depth $e$,
has a frontier variable at depth $\geq e-1\geq d_\beta$.

This concludes the proof that $\Inter_{i-1}$ contains all atoms 
that were  inferred at a chase step before $\vec{s}[i]$ and use terms in $\fLineagePath_B(\vec{s}[i])$.
By (1)--(3) above, this completes the proof.
\end{proof}

\theoremExpSpaceCompleteness*

\begin{proof}
	Membership follows from the correctness (Lemma~\ref{lemma_algo_pg_chase_correct}) and completeness (Lemmas~\ref{lemma_algo_pg_chase_complexity},
	\ref{lemma_nchase_short_sequences}, and \ref{lemma_nchase_single_atom}) of the tree-based chase, where the algorithm follows
	$|q|$ step sequences based on the $|q|$ atoms in a particular query match to materialise the match in the $|q|$ iterations of the outer loop in \lcRef{l_chase_qloop}.
	Algorithm~\ref{algo_pg_chase} therefore decides query entailment in $\kNExpSpace{(\kappa-1)}$, and we get membership in $\kExpSpace{(\kappa-1)}$ by Savitch's Theorem.
	
	We show hardness via reduction from the word problem for alternating Turing Machines (ATMs) with a $(\kappa-1)$-exponential time bound, which is $\kExpSpace{(\kappa-1)}$-complete \cite{ATM}.
	
	Let $\mathcal{M} = \tuple{Q, \Gamma, \Delta, q_0, t}$ be an ATM with a finite set of states $Q$, a finite tape alphabet $\Gamma$ consisting of an input alphabet and a special symbol $\blank$ (blank), a transition relation $\Delta \subseteq (Q\times \Gamma) \times (Q\times \Gamma\times \{l, r\})$, an initial state $q_0 \in Q$, and a function $t: Q \to \{ \exists, \forall, \opfont{acc} \}$ that marks states as existential, universal, or accepting.
	As usual, we write configurations as $w_l q w_r$ where $w_l$ denotes the tape symbols left of the read-write head,
	$q\in Q$ is the current state, and $w_r$ is a sequence of symbols to the right of the read-write head (the first symbol of which is underneath the head).
	Tape symbols right of $w_r$ are assumed to be $\blank$.
	For a configuration $C = w_l q \sigma w_r$ and a transition $\tuple{\tuple{q, \sigma}, \tuple{q^+, \sigma^+, d}}\in\Delta$, there is a successor configuration
	\begin{align*}
		C^+ \defeq
		\begin{cases} 
			w_l \sigma^+ q^+ w_r\blank & \text{if $d = r$,} \\
			\hat{w}_l q^+ \sigma_l \sigma^+ w_r & \text{if $d = l$ and $w_l = \hat{w}_l \sigma_l$ for some $\sigma_l \in \Gamma$.}
		\end{cases}
	\end{align*}
	A configuration $w_l q w_r$ is accepting if
	\begin{enumerate*}
		\item $t(q) = \opfont{acc}$,
		\item $t(q) = \exists$ and there is an accepting successor configuration, or
		\item $t(q) = \forall$ and all successor configurations are accepting.
	\end{enumerate*}
	$\mathcal{M}$ accepts a word $w$ if the initial configuration $q_0w$ is accepting.
	
	For a word $w = \sigma_1\ldots\sigma_n$ with $n>0$ (the case of the empty word is irrelevant for hardness),
	let $\Dnter_w$ denote the database with the facts
	$\pnFirst_0(1)$, $\pnNext_0(i, i+1)$ for $i\in\{1,\ldots,n-1\}$, $\pnLast_0(n)$,
	and $\pnSymb(i,\sigma_i)$ for $i\in\{1,\ldots,n\}$.
	Let $\aprogram_\leq$ denote a set of tgds constructed as in the \hyperref[proof_thmRankCompl]{proof of Theorem~\ref*{theo_rank_compl}}
	to entail a $(\kappa-1)$-exponentially long chain over the initial chain encoded in predicates $\pnFirst_0$, $\pnNext_0$, and $\pnLast_0$.
	We assume that the constructed $(\kappa-1)$-exponential chain is encoded predicates $\pnFirst$, $\pnNext$, and $\pnLast$.
	Moreover, we extend $\aprogram_\leq$ with the following rules to encode the transitive (non-reflexive) closure of $\pnNext$:
	\begin{align}
	\pnNext(x,y) & \to \pnNext^+(x,y)\\
	\pnNext(x,y)\wedge \pnNext^+(y,z)& \to \pnNext^+(x,z)
	\end{align}

	We use facts $\pnStep(c,i)$ to encode that $c$ is an $i$th configuration in a run, i.e., there is a sequence $c_0, \ldots, c_i = c$ with $c_0$ being the initial configuration,
	facts $\pnState(c,q)$ to encode the state $q$ of configuration $c$,
	facts $\pnHPos(c,i)$ to encode that the head of $\mathcal{M}$ is at the $i$th position of the tape in configuration $w$, and
	facts $\pnTape(c,i,s)$ to encode that the $i$th position of the tape of configuration $c$ is symbol $s$.
	The following tgds $\aprogram_{\opfont{init}}$ then establish the initial configuration $c_0$ (with constants $c_0$, $q_0$, and $\blank$):
	\begin{align}
		\pnFirst(i) &\to \pnStep(c_0,i) \land \pnState(c_0, q_0) \land \pnHPos(c_0,i)\label{eq_atm_init_conf}\\
		\pnFirst(i) \land \pnFirst_0(j) &\to \predname{samepos}(i,j)  \label{eq_atm_match_chains_base}\\
		\predname{samepos}(i,j)\land\pnNext(i,i^+)\land\pnNext_0(j,j^+) &\to\predname{samepos}(i^+,j^+) \label{eq_atm_match_chains_step}\\
		\predname{samepos}(i,j)\land\pnSymb(j,s) &\to\pnTape(c_0,i,s) \label{eq_atm_init_word}\\
		\predname{samepos}(i,j)\land\pnLast_0(j)\land\pnNext^+(i,i^+) & \to \pnTape(c_0,i,\blank) \label{eq_atm_init_blanks}
	\end{align}
	The tgd \eqref{eq_atm_init_conf} initialises step, state, and head position of the initial configuration $c_0$.
	The tgds \eqref{eq_atm_match_chains_base} and \eqref{eq_atm_match_chains_step} match elements of the initial chain
	to the first elements of the $(\kappa-1)$-exponential chain. This is then used to transcribe the input word to the
	initial tape \eqref{eq_atm_init_word}, with the remaining tape filled with blanks \eqref{eq_atm_init_blanks}.

	Next, we encode how to construct a null for the successor of a configuration.
	For this purpose, we introduce predicates $\pnNextConf_\delta$ for every $\delta\in\Delta$, each encoding successor configurations 
	for this transition in a certain step.
	Concretely, for all $\delta,\delta'\in\Delta$ with $\delta = \tuple{\tuple{q, \sigma}, \tuple{q^+, \sigma^+, d}}$, the tgds $\aprogram_{+}$ contain
	the following tgds (with constants $q$ and $\sigma$):
	\begin{align}
		\pnStep(C,i) \land \pnNext(i,i^+) \land \pnHPos(C,p) \land
			\pnTape(C,p,\sigma) \land \pnState(C,q) &\to \exists v_\delta \ldotp \pnNextConf_\delta(i^+,C,v_\delta)\label{eq_atm_conf_get}\\
		\pnNextConf_\delta(i,C^-,C) &\to \pnNextConf_{\delta'}(i,C,C) \label{eq_atm_conf_base_prop}\\
		\pnNextConf_\delta(i,C^-,C) \land \pnNextConf_{\delta'}(j,C^-,C^-) &\to \pnNextConf_{\delta'}(j,C,C) \label{eq_atm_conf_step_prop}
	\end{align}
	The tgd~\eqref{eq_atm_conf_get} creates a new null for each valid transition $\delta$ from a configuration $C$ with state $q$ and $s$ at the head position.
	Facts of the form $\pnNextConf_{\delta}(i^+,C,C^+)$ encode that $C^+$ is a $\delta$-successor of $C$ at depth $i^+$.
	The idea is that each $i$ of the linear order encoded by $\pnNext$ can be used at most once per configuration path $c_1, \ldots, c_n$ to create a successor.
	
	$\aprogram_+$ will induces a saturating strongly connected component in the overall labeled dependency graph.
	Indeed, Definition~\ref{def_esaturating} is satisfied for the set $E = \{ v_\delta \ledgeto{C(\delta')} v_\delta' \mid \delta, \delta' \in \Delta \}$,
	where we use $C(\delta')$ to disambiguate the universal variables $C$ in tgds of the form \eqref{eq_atm_conf_get}.
	Note that the empty path is the only $\bar{E}$-path in this case.
	Now, for tgd $\arule_{\delta'}$ with existential variable $v_{\delta'}$, tgds~\eqref{eq_atm_conf_base_prop} ensure base-propagation for $e = v_\delta' \ledgeto{C(\delta')} v_{\delta'} \in E$ and tgds~\eqref{eq_atm_conf_step_prop} ensure step-propagation for $e_1, e_2 = v_\delta' \ledgeto{C(\delta')} v_{\delta'}  \in E$.

	Next, the tgds $\aprogram_{\to}$ encode the $\delta$-successor $C^+$ of a configuration $C$ based on the encoding of $C$.
	For $\delta = \tuple{\tuple{q, \sigma}, \tuple{q^+, \sigma^+, d}}$, $\aprogram_{\to}$ contains the following tgds:
	\begin{align}
		\begin{split}
		\pnStep(C,i) \land{} \pnNext(i,i^+) \land \pnNextConf_\delta(i^+,C,C^+)
			& \to \pnStep(C^+, i^+) \land \pnNextConfReal_\delta(C,C^+) \land \pnState(C^+,q^+)
		\end{split}\label{eq_atm_trans_step}\\
		\begin{split}
		\pnHPos(C,p) \land \pnNext^+(p,p') \land{} &\pnTape(C,p',s') \land \pnNextConfReal_\delta(C,C^+) \\
			&\to \pnTape(C^+,p',s')
		\end{split}\label{eq_atm_trans_tape_left}\\
		\begin{split}
		\pnHPos(C,p) \land \pnNext^+(p',p) \land{} &\pnTape(C,p',s') \land \pnNextConfReal_\delta(C,C^+) \\
			&\to \pnTape(C^+,p',s')
		\end{split}\label{eq_atm_trans_tape_right}\\
		\pnHPos(C,p) \land{} \pnNextConfReal_\delta(C,C^+) &\to \pnTape(C^+,p,\sigma^+)\label{eq_atm_trans_tape_head}
	\end{align}
	Here, the tgd~\eqref{eq_atm_trans_step} writes the step and state of $C^+$ and introduces a notion $\pnNextConfReal_{\delta}(C,C^+)$, which states that $C^+$ is the `real' $\delta$-successor of $C$ (tgds~\eqref{eq_atm_conf_base_prop} and \eqref{eq_atm_conf_step_prop} makes $C^+$ a pseudo-successors of itself for all $i \leq i^+$).
	The tgds~\eqref{eq_atm_trans_tape_left}--\eqref{eq_atm_trans_tape_head} write the tape of configuration $C^+$.
	Moreover, $\aprogram_{\to}$ requires tgds to update the head position, and we make a case distinction based on the direction $d$ to which the head moves.
	If $d = l$, $\aprogram_{\to}$ contains the tgd:
	\begin{align}
		\pnHPos(C,p) \land \pnNext(p',p) \land \pnNextConfReal_\delta(C,C^+) &\to \pnHPos(C^+,p') %
	\end{align}
	and, otherwise, if $d = r$, $\aprogram_{\to}$ contains the tgd:
		\begin{align}
		\pnHPos(C,p) \land{} &\pnNext(p,p') \land \pnNextConf_\delta(C,C^+) \to \pnHPos(C^+,p')
	\end{align}

	Finally, we define tgds to evaluate which configurations are accepting.
	The set of tgds $\aprogram_{\opfont{eval}}$ contains the following tgds, where we introduce further 
	predicates $\pnAccept$ and $\pnAccept_\delta$ for each $\delta\in\Delta$. %
	\begin{align}
	\intertext{1. For every state $q_a\in Q$ with $t(q_a) = \opfont{acc}$:}
		\pnState(C, q_a) &\to \pnAccept(C) \label{eq_atm_eval_q_accept}\\
	\intertext{2. For every state $q_\exists\in Q$ with $t(q_\exists) = \exists$, and every transition $\delta = \tuple{\tuple{q_\exists,\sigma}, \tuple{q^+,\sigma^+,d}} \in \Delta$:}
		\pnNextConfReal_\delta(C,C^+) \land \pnAccept(C^+) &\to \pnAccept(C) \label{eq_atm_eval_q_exists}\\
	\intertext{3. For every state $q_\forall\in Q$ with $t(q_\forall) = \forall$, and every symbol $\sigma\in\Gamma$ with $\delta_1,\ldots,\delta_n$ a list of all transitions
	of form $\tuple{\tuple{q_\forall,\sigma}, \tuple{q^+,\sigma^+,d}} \in \Delta$:}
		\pnNextConfReal_{\delta_1}(C,C^+) \land \pnAccept(C^+) &\to \pnAccept_{\delta_1}(C) \label{eq_atm_eval_q_forall_first}\\
		\pnAccept_{\delta_i}(C) \land \pnNextConfReal_{\delta_{i+1}}(C,C^+) \land \pnAccept(C^+) &\to \pnAccept_{\delta_{i+1}}(C) \label{eq_atm_eval_q_forall_step}\\
		\pnAccept_{\delta_i}(C) &\to \pnAccept(C)\label{eq_atm_eval_q_forall_last}
	\end{align}
	The tgds recursively mark configurations as accepting.
	Concretely, \eqref{eq_atm_eval_q_accept} directly marks all configuration with an accepting state, and
	\eqref{eq_atm_eval_q_exists} marks configurations with an existential state if they have an accepting successor.
	The tgds~\eqref{eq_atm_eval_q_forall_first}--\eqref{eq_atm_eval_q_forall_last} mark a configuration with a universal state if all of its successors are accepting.
	To achieve this, the successors are traversed in an arbitrary but fixed order $\delta_1, \ldots, \delta_n$.
	Semantically, one could just combine these rules into one, but the above splitting ensures path-guardedness.
	Indeed, the relation $\preceq$ of Definition~\ref{def_constraints_new} for the above sets of tgds contains
	$\ppos{\pnNextConf_{\delta}}{2} \preceq \ppos{\pnNextConf_{\delta}}{3}$ and $\ppos{\pnNextConfReal_{\delta}}{1} \preceq \ppos{\pnNextConfReal_{\delta}}{2}$.
	Therefore, all tgds are path-guarded.

	To conclude, let $\aprogram = \aprogram_\leq \cup \aprogram_{\opfont{init}} \cup \aprogram_{+} \cup \aprogram_{\to} \cup \aprogram_{\opfont{eval}}$.
	As $\aprogram$ constructs and evaluates the configuration tree of $\mathcal{M}$, we obtain that $\mathcal{M}$ accepts $w$ if and only if $\fnChase{\aprogram, \Dnter_w} \models \pnAccept(c_0)$.
	Since $\aprogram$ is saturating with $\rank{\aprogram} = \kappa$, arboreous, and path-guarded, this shows the claim.
\end{proof}
 
\end{document}